\documentclass[11pt,a4paper]{amsart}

\usepackage[dvipsnames]{xcolor}
\usepackage[colorlinks,allcolors=MidnightBlue]{hyperref}
\usepackage[margin=1in]{geometry}

\usepackage[libertine]{newtx}
\usepackage{inconsolata}

\usepackage[T1]{fontenc}
\usepackage{amsmath,amsthm,mathtools}
\usepackage{thmtools}
\usepackage{tablefootnote}
\usepackage{cite}
\usepackage{hyperref}
\usepackage{nameref}
\usepackage{crossreftools}
\usepackage{microtype}
\usepackage{anyfontsize}
\usepackage{comment}
\usepackage{csquotes}

\DeclareUnicodeCharacter{0301}{\'a}
\DeclareUnicodeCharacter{2212}{-}
\DeclareUnicodeCharacter{03B5}{\ensuremath{\varepsilon}}

\usepackage{textcomp}
\usepackage{newtxmath}
\usepackage[scr=rsfso]{mathalfa}
\usepackage{bm}

\usepackage[caps]{complexity}

\usepackage[textwidth=2cm]{todonotes}
\renewcommand\emptyset\varnothing{}
\renewcommand\epsilon\varepsilon{}

\DeclarePairedDelimiter\floor{\lfloor}{\rfloor}
\DeclarePairedDelimiter\ceil{\lceil}{\rceil}

\let\E\relax
\DeclareMathOperator*{\E}{\mathbb{E}}
\DeclareMathOperator{\Forb}{Forb}
\DeclareMathOperator{\ar}{ar}
\DeclareMathOperator{\Hom}{Hom}

\renewcommand{\bar}[1]{\overline{#1}}

\DeclareMathOperator{\MMSNP}{MMSNP}
\DeclareMathOperator{\PMMSNP}{PMMSNP}
\DeclareMathOperator{\CSP}{CSP}
\DeclareMathOperator{\PCSP}{PCSP}

\newcommand{\family}[1]{\mathcal{#1}}
\newcommand{\templ}[2]{\family{F}^{(#1)}_{#2}}
\newcommand{\urel}[3]{R^{(#1)}_{#2:#3}}
\newcommand{\edgerel}[3]{Q^{(#1)}_{#2:#3}}

\newcommand{\tuple}[1]{\bar{#1}}

\newcommand{\yes}{\textsf{Yes}}
\newcommand{\no}{\textsf{No}}

\newcommand{\balanced}{N\v{Z}-balanced}

\newcommand{\kinl}[2]{\ensuremath{#1\textrm{-}\mathsf{in}\textrm{-}#2}}

\newcommand{\transpose}[2]{({#1} \mid \cdots \mid {#2})^T}

\usepackage{xparse}
\ExplSyntaxOn{}

\newcommand{\LO}{\mathsf{LO}}
\newcommand{\opNAE}{\mathsf{NAE}}
\NewDocumentCommand{\NAE}{o m}
 {
  \ensuremath{
    \IfNoValueTF{#1}
      {\opNAE\sb{#2}}
      {\opNAE\sb{#2}\sp{(#1)}}
  }
 }

\newcommand{\opInf}{\mathbf{Inf}}
\NewDocumentCommand{\Inf}{o >{\SplitArgument{1}{,}}m}
 {
   \__Inf:nnn { #1 } #2
 }
\cs_new:Nn \__Inf:nnn
 {
   \ensuremath{
     \IfNoValueTF{#1}
        {{\opInf\sb{#3}}{\left[#2\right]}}
        {{\opInf\sb{#3}}{\left[#2; #1\right]}}
   }
 }
\NewDocumentCommand{\lowInf}{o >{\SplitArgument{2}{,}}m}
 {
   \__lowInf:nnnn { #1 } #2
 }
\cs_new:Nn \__lowInf:nnnn
 {
   \ensuremath{
     \IfNoValueTF{#1}
        {{\opInf\sb{#3}\sp{\le #4}}{\left[#2\right]}}
        {{\opInf\sb{#3}\sp{\le #4}}{\left[#2; #1\right]}}
   }
 }
\NewDocumentCommand{\noisyInf}{o >{\SplitArgument{2}{,}}m}
 {
   \__noisyInf:nnnn { #1 } #2
 }
\cs_new:Nn \__noisyInf:nnnn
 {
   \ensuremath{
     \IfNoValueTF{#1}
        {{\opInf\sb{#3}\sp{(#4)}}{\left[#2\right]}}
        {{\opInf\sb{#3}\sp{(#4)}}{\left[#2; #1\right]}}
   }
 }

\ExplSyntaxOff{}

\usepackage[capitalize,nameinlink,noabbrev]{cleveref}
\usepackage{thm-restate}

\usepackage{tikz}
\usetikzlibrary{calc}

\theoremstyle{plain}
\newtheorem{theorem}{Theorem}[section]
\newtheorem{lemma}[theorem]{Lemma}
\newtheorem{proposition}[theorem]{Proposition}
\newtheorem{corollary}[theorem]{Corollary}
\newtheorem{conjecture}[theorem]{Conjecture}

\theoremstyle{definition}
\newtheorem{definition}[theorem]{Definition}
\newtheorem{example}[theorem]{Example}

\theoremstyle{remark}
\newtheorem{remark}[theorem]{Remark}
\newtheorem{observation}[theorem]{Observation}
\newtheorem{claim}[theorem]{Claim}

\title{Towards infinite PCSP:\@ a dichotomy for monochromatic cliques}\thanks{The conference version of this article appeared in the Proceedings of the 41st Annual Symposium on Logic in Computer Science (LICS 2026).}

\author[D. Banakh]{Demian {Banakh}}
\address{Faculty of Mathematics and Computer Science, Jagiellonian University, Krak{\'o}w, Poland \and
Doctoral School of Exact and Natural Sciences, Jagiellonian University, Krak{\'o}w, Poland}
\email{demian.banakh@doctoral.uj.edu.pl}
\urladdr{https://dembanakh.github.io/}
\thanks{The first author was supported by the National Science Centre,
Poland under the Weave-UNISONO call in the Weave programme 2021/03/Y/ST6/00171
and the Polish National Agency for Academic Exchange under Zawacka programme.
For the purpose of Open Access, the author has applied a CC-BY public copyright
licence to any Author Accepted Manuscript (AAM) version arising from this
submission.}

\author[A. Barsukov]{Alexey {Barsukov}}%
\email{barsukov1995@gmail.com}
\urladdr{https://abarsukov.github.io/}
\thanks{The second author was funded by the European Union (ERC, POCOCOP, 101071674). Views and opinions expressed are however those of the author(s) only and do not necessarily reflect those of the European Union or the European Research Council Executive Agency. Neither the European Union nor the granting authority can be held responsible for them.}

\author[T.-V. Nakajima]{Tamio-Vesa {Nakajima}}%
\address{Faculty of Mathematics and Informatics, Philipps University, Marburg, Hesse, Germany}
\email{nakajima@uni-marburg.de}
\urladdr{https://tamionv.ro/}

\begin{document}

\begin{abstract}
The logic MMSNP is a well-studied fragment of Existential Second-Order logic
that, from a computational perspective, captures finite-domain Constraint
Satisfaction Problems (CSPs) modulo polynomial-time reductions. At the same
time, MMSNP contains many problems that are expressible as $\omega$-categorical
CSPs but not as finite-domain ones.

We initiate the study of \emph{Promise} MMSNP (PMMSNP), a promise analogue of
MMSNP.\@ We show that every PMMSNP problem is poly-time equivalent to a
(finite-domain) \emph{Promise} CSP (PCSP), thereby extending the classical
MMSNP--CSP correspondence to the promise setting. We then investigate the 
complexity of PMMSNPs arising from forbidding monochromatic cliques, a class
encompassing promise graph colouring problems. For this class, we obtain a full
complexity classification conditional on the \emph{Rich 2-to-1 Conjecture}, a
recently proposed perfect-completeness surrogate of the Unique Games Conjecture.

As a key intermediate step which may be of independent interest, we prove that
it is \NP-hard, under the Rich 2-to-1 Conjecture, to properly colour a uniform
hypergraph even if it is promised to admit a colouring satisfying a certain
technical condition called \emph{reconfigurability}. This proof is an extension
of the recent work of Braverman, Khot, Lifshitz and Minzer (Adv.~Math.~2025). To
illustrate the broad applicability of this theorem, we show that it implies most
of the linearly-ordered colouring conjecture of Barto, Battistelli, and Berg
(STACS 2021).
\end{abstract}

\maketitle

\section{Introduction}\label{sec:intro}
Monotone Monadic SNP (MMSNP) is a large fragment of Existential Second-Order
Logic (ESO) introduced by Feder and Vardi~\cite{federvardi1998}, with the goal
of having a \P{} vs.~\NP{}-complete dichotomy (whereas the whole logic ESO
contains \NP{}-intermediate problems assuming $\P\neq\NP$~\cite{fagin,ladner}).
Fixing a language $\sigma$ consisting of relational symbols $R_1, \ldots, R_k$,
an MMSNP formula $\Phi$ is an existential second-order formula of the form
\[
    \Phi = \exists M_1 \dots M_c  \forall x_1 \ldots x_n\  \phi(x_1, \ldots, x_n),
\]
which additionally satisfies the following syntactic restrictions:
\begin{description}
    \item[Monotone] every $\sigma$-atom $R_i(\dots)$ within $\phi$ has an odd
        number of nested negations before it.
    \item[Monadic] every relation among $M_1, \ldots, M_c$ is unary i.e.~it is
        a set.
    \item[No inequality] the quantifier-free part $\phi$ does not contain $=$
        or $\neq$.
\end{description}
\begin{example}\label{ex:two_colouring}
The sentence
\[
    \Phi = \exists A, B ~ \forall x, y~  
    \begin{pmatrix}
        (A(x) \lor B(x)) \\ \land \\ (\lnot A(x) \lor \lnot B(x)) \\ \land \\
        (A(x) \land A(y)) \Rightarrow \lnot E(x, y) \\ \land \\ (B(x) \land B(y)) \Rightarrow \lnot E(x, y)
    \end{pmatrix}
\]
is in MMSNP.\@ In this case, $\sigma=\{E\}$ has one binary relation symbol, so
relational $\sigma$-structures are directed graphs. For a digraph $\mathbb G$,
note that $\mathbb G \vDash \Phi$ (read $\mathbb G$ models $\Phi$, or $\Phi$ is
satisfied in $\mathbb G$) if and only if we can partition the vertex set $G$ of
$\mathbb G$ into two sets $A, B$ so that there is no edge between vertices from
the same set i.e.~when $\mathbb G$ is 2-colourable.
\end{example}

As illustrated by \cref{ex:two_colouring}, there is a natural relation between
MMSNP satisfiability and \emph{Constraint Satisfaction Problems} (CSPs). For a
relational $\sigma$-structure $\mathbb A = (A; R_1^\mathbb A, \ldots,
R_k^\mathbb A)$, $\CSP(\mathbb A)$ is a computational problem, where we are
given a finite $\sigma$-structure $\mathbb I = (I; R_1^\mathbb I, \ldots,
R_k^\mathbb I)$ and must decide the existence of a mapping $h\colon I\to A$ that
preserves every relation $R_i\in\sigma$. Such a mapping is called a
\emph{homomorphism} and is denoted by $h\colon \mathbb I\to \mathbb A$. Whenever
$|A| < \infty$, we are able to construct, similarly to \cref{ex:two_colouring},
a sentence $\Phi$ in MMSNP such that $\mathbb I\to\mathbb A$ if and only if
$\mathbb I\vDash \Phi$, for every input $\mathbb I$.

Feder and Vardi showed that, although MMSNP is strictly more expressive than
finite-domain CSP, for every MMSNP sentence $\Phi$ there exists a finite
relational structure $\mathbb{A}$ such that $\MMSNP(\Phi)$ is equivalent to
$\CSP(\mathbb{A})$ under randomised poly-time reductions; they also conjectured
a \P{} vs.~\NP{}-complete dichotomy for both classes. These reductions were
later derandomised by Kun~\cite{kun2013}.

The conjecture of Feder and Vardi about CSPs inspired an active research
programme. It culminated two decades later, when Bulatov and Zhuk independently
proved that CSPs admit a $\P$ vs.~$\NP$-complete dichotomy, using algebraic
tools~\cite{bulatov,zhuk_conference,zhuk}. This automatically implied a similar
dichotomy for $\MMSNP$s.\ Shortly thereafter, Bodirsky, Madelaine, and
Mottet~\cite{bodirsky_madelaine_mottet_conf,bodirsky_madelaine_mottet} showed
that the complexity of MMSNP problems also depends on the algebraic properties
of corresponding ($\omega$-categorical) CSPs, thus confirming the conjecture of
Bodirsky and Pinsker~\cite{bodirsky_pinsker_conjecture} on MMSNP problems.\@

\subsection*{The logic of forbidden patterns}

Despite the fact that MMSNP is defined in terms of syntactic restrictions
applied to ESO sentences, it has an equivalent ``combinatorial'' definition
called \emph{forbidden pattern problems}~\cite{madelaine_forbidden_patterns}.
Each such problem is described by a finite family $\mathcal F$ of
(element-)coloured connected relational structures $(\mathbb F,f)$ called
\emph{forbidden patterns}, where the mapping $f\colon F\to[c]$ stands for a
colouring of the elements of $\mathbb F$ with  $c$-many colours, for some
$c\in\mathbb N$. For a given relational structure $\mathbb I$, the question is
to decide whether the elements of $\mathbb I$ are colourable in such a way that
there is no (colour respecting) homomorphism from a  forbidden pattern of
$\mathcal F$. A coloured structure that avoids a homomorphism from a forbidden
pattern to itself is called \emph{$\mathcal F$-free}. For example, the problem
of 2-colouring a graph so as to avoid all monochromatic triangles belongs to
MMSNP;\@ a drawing of the family $\mathcal{F}$ for this problem can be found in
\cref{fig:no_mono_tri}. Every forbidden pattern problem is clearly in MMSNP and
every problem in MMSNP is a finite union of forbidden pattern
problems~\cite[Corollary 13]{madelaine_forbidden_patterns}. Taking into account
this equivalence and that the problems studied in this paper are much easier to
formulate in the framework of forbidden patterns, we will stick to this notation
further on. So, we also denote by $\MMSNP(\mathcal{F})$ the forbidden pattern
problem associated with $\mathcal{F}$.

\subsection*{MMSNP and finite CSP}

As mentioned above, Feder and Vardi~\cite{federvardi1998} demonstrated that each
MMSNP is equivalent to some finite-domain CSP.\@ In fact, their construction is
quite straightforward. They define the CSP template as follows: the domain is
the set $[c]$ of colours used in $\mathcal F$; every forbidden pattern $(\mathbb
F, f) \in \mathcal F$ yields an $|F|$-ary relation in the template that consists
of all colourings $g\colon F\to [c]$ that do not belong to $\mathcal F$. The
reduction from $\MMSNP(\mathcal F)$ to this CSP then follows in a natural way:
given $\mathbb I$, we produce $\mathbb I'$ --- an instance of our CSP --- by
listing, for each forbidden pattern $(\mathbb F, f) \in \mathcal F$, all
homomorphisms from $\mathbb F$ to $\mathbb I$.

\begin{example}[No monochromatic triangle]\label{ex:no_mono_tri}
Suppose that there are just two colours: $b,r$; and that $\templ{2}{3}$ consists
of two copies of $\mathbb K_3$ --- one is coloured fully in $b$, the other fully
in $r$ (cf.~\cref{fig:no_mono_tri}). The $\sigma$-structure associated with
$\templ{2}{3}$ has domain $\{b,r\}$ and one symmetric ternary relation
${\{b,r\}}^3\setminus \{(b,b,b),(r,r,r)\}$; thus the corresponding CSP template
is just the ternary \emph{Not All Equal} relation ($\NAE{3}$). The reduction
from our MMSNP to the corresponding CSP simply applies this constraint to every
triangle (and every loop) in the input graph.
\end{example}

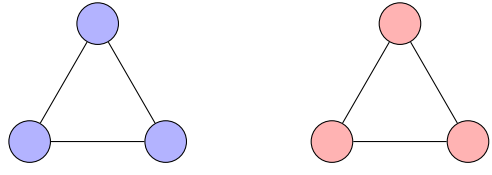
\begin{figure}
    \centering
    \begin{tikzpicture}[
      scale=1,
      every node/.style={circle, draw, minimum size=5.5mm, inner sep=0pt}
    ]

      \def\s{1.8}
      \pgfmathsetmacro{\h}{\s*sqrt(3)/2}

      \node[fill=MidnightBlue!30] (A1) at (0,0) {};
      \node[fill=MidnightBlue!30] (A2) at (\s,0) {};
      \node[fill=MidnightBlue!30] (A3) at ({\s/2},\h) {};

      \draw (A1) -- (A2) -- (A3) -- (A1);

      \node[fill=BrickRed!30] (B1) at (4,0) {};
      \node[fill=BrickRed!30] (B2) at ({4+\s},0) {};
      \node[fill=BrickRed!30] (B3) at ({4+\s/2},\h) {};

      \draw (B1) -- (B2) -- (B3) -- (B1);

    \end{tikzpicture}
    \caption{Template $\mathcal F_3^{(2)}$ for 2-colouring a graph avoiding monochromatic triangles.}\label{fig:no_mono_tri}
\end{figure}

The reduction in the opposite direction relies on the \emph{Sparse
Incomparability Lemma} which states that, for every finite $\sigma$-structure
$\mathbb B$ and $\ell >1$, the complexity of $\CSP(\mathbb B)$ remains the same
even after restricting the input to structures without cycles of length $<\ell$.

\subsection*{MMSNP and infinite CSP}

Some sentences in MMSNP describe problems that are not finite-domain CSPs: for
example, the \emph{triangle-free} property can be expressed in MMSNP but there
are triangle-free graphs with arbitrarily high chromatic number. (Hence no
finite CSP has as \yes{} instances precisely the set of finite triangle-free
graphs.) The whole class of CSPs of $\omega$-categorical structures is too
large: it contains even coNP-complete and coNP-intermediate
problems~\cite{bodirsky_grohe}.\footnote{However, the question whether there are
$\omega$-categorical structures whose CSP is NP-intermediate is still
open~\cite{BookBodirsky}.} On the other hand, many natural families of
infinite-domain CSPs exhibit a \P{} vs.~\NP-complete dichotomy: reducts of
$(\mathbb Q,<)$~\cite{temporal}, reducts of homogeneous
graphs~\cite{reducts_graphs}, reducts of the random poset~\cite{reducts_posets},
reducts of the random tournament~\cite{reducts_tournament}. All these classes
fall into the scope of the family of \emph{reducts of finitely bounded
homogeneous structures}.

For CSPs of these structures, Bodirsky and Pinsker conjectured a
\P~vs.~\NP{}-complete dichotomy with a similar universal algebraic
characterisation as for finite-domain CSPs,
see~\cite{bodirsky_pinsker_conjecture}. The name for this class of structures
comes from the way they are created. It starts from a \emph{finite} family
$\mathcal F$ of structures, called the \emph{bounds} and the class
$\Forb^\mathrm{emb}(\mathcal F)$ of finite structures that have no substructure
from $\mathcal F$. The class $\Forb^\mathrm{emb}(\mathcal F)$ is supposed to
have the \emph{amalgamation property}: for every $\mathbb B_1,\mathbb B_2$ in
$\Forb^\mathrm{emb}(\mathcal F)$ that have a common substructure $\mathbb A$
witnessed by embeddings $e_i\colon \mathbb A \hookrightarrow \mathbb B_i$ (for
$i\in\{1,2\}$), there is a structure $\mathbb C$ in $\Forb^\mathrm{emb}(\mathcal
F)$ and embeddings $f_i\colon \mathbb B_i\hookrightarrow \mathbb C$ (for
$i\in\{1,2\}$) such that $f_1\circ e_1 = f_2\circ e_2$. Then, using
Fra\"{\i}ss\'e's Theorem~\cite[Section 6.1]{hodges}, we inductively construct a
countably infinite structure $\mathbb H$ (called the \emph{Fra\"{\i}ss\'e-limit}
of $\Forb^\mathrm{emb}(\mathcal F)$), whose finite substructures are exactly
those from $\Forb^\mathrm{emb}(\mathcal F)$ and which is \emph{homogeneous}:
every isomorphism between finite substructures of $\mathbb H$ extends to an
automorphism of $\mathbb H$. The final structure, called a (first-order)
\emph{reduct} of $\mathbb H$, is obtained from $\mathbb H$ by adding new
first-order definable relations and then forgetting some relations. Note that
structures whose CSPs capture MMSNP problems are constructed similarly: for
instance, the class $\Forb^\mathrm{emb}(\mathcal F_3^{(2)})$ of coloured graphs
without monochromatic triangles (\cref{fig:no_mono_tri}) already has the
amalgamation property. So, it is not surprising that every MMSNP problem is
equivalent to a finite union of CSPs of finitely bounded homogeneous
structures~\cite{bodirsky_madelaine_mottet_conf,bodirsky_madelaine_mottet}, and,
as it was mentioned, a dichotomy for them is the special case of the conjecture
of Bodirsky and Pinsker.

In spite of this natural connection of MMSNP to infinite-domain CSPs, our
techniques in this paper will be concerned solely with finite-domain relational
structures.

\subsection*{Recolouring and containment}

In their paper~\cite{federvardi1998}, Feder and Vardi showed that containment
for problems in MMSNP is decidable. (By containment, we mean deciding whether
all \yes{} instances of one problem are also \yes{} instances of the other
problem.) For two families $\mathcal F_1$ and $\mathcal F_2$ of forbidden
patterns, the containment (of the corresponding MMSNP problems) is characterised
by the existence of a mapping $r\colon[c_1]\to[c_2]$ between the sets of colours
that lifts up to a mapping from $\Forb^\mathrm{hom}(\mathcal F_1)$ to
$\Forb^\mathrm{hom}(\mathcal F_2)$~\cite{madelaine_containment}, where
$\Forb^\mathrm{hom}(\mathcal F)$ is the set of finite vertex-coloured structures
that do not admit a homomorphism from any member of $\mathcal F$. Such a mapping
is called a \emph{recolouring}. We write $\mathcal F_1 \to \mathcal F_2$ if a
recolouring from $\mathcal F_1$ to $\mathcal F_2$ exists.

\subsection*{Promise MMSNP}

In this paper, we initiate the study of a promise problem variant of MMSNP,
inspired by the development of \emph{Promise CSP}. Namely, a \emph{Promise
MMSNP} (PMMSNP) over a pair of MMSNP templates $\mathcal F, \mathcal G$ such
that $\mathcal F \to \mathcal G$ asks to distinguish whether the input structure
admits an $\mathcal F$-free colouring, or not even a $\mathcal G$-free one. We
let $\PMMSNP(\mathcal F, \mathcal G)$ denote this problem. Note that the fact
that $\mathcal F \to \mathcal G$ implies that the \yes{} and \no{} instances of
this problem are disjoint. Taking into account that every MMSNP is a union of
$\omega$-categorical CSPs, our research can be considered as the first step
towards understanding the complexity of \emph{Promise CSPs} of countably
infinite structures.

Every finite-domain CSP can be encoded with an MMSNP sentence, thus every finite-domain Promise CSP is equal to some Promise MMSNP.
In fact, the converse holds as well, modulo polynomial-time reductions. That is, we
prove that every Promise MMSNP is poly-time equivalent to a Promise CSP
defined by a suitable pair of finite structures. Recall that for a promise
structure $\mathbb A$ and a target structure $\mathbb B$ such that $\mathbb A
\to \mathbb B$, $\PCSP(\mathbb A, \mathbb B)$ is a problem in which, given a
structure $\mathbb I$, the goal is to distinguish whether $\mathbb I \to \mathbb
A$ or $\mathbb I \not\to \mathbb B$. Here as well, the fact that $\mathbb A \to
\mathbb B$ ensures that the \yes{} and \no{} instances are disjoint.

\begin{restatable}{theorem}{mainequivalence}\label{th:mainequivalence}
    If $\mathcal F \to \mathcal G$, then there exist finite structures $\mathbb
    S \to \mathbb T$ such that $\PMMSNP(\mathcal F, \mathcal G)$ and
    $\PCSP(\mathbb S, \mathbb T)$ are poly-time equivalent.
\end{restatable}

This is a direct analogue of the MMSNP--CSP correspondence of Feder and Vardi in
the promise realm. The PCSP complexity landscape remains largely unresolved,
despite significant progress, and even simple subclasses, like Boolean PCSPs or
Approximate Graph Homomorphism,\footnote{In the recent breakthrough result, Fei, Minzer, and Wang \cite{dto1proof} claim to have completely resolved the computational complexity of the Approximate Graph Colouring.} are not known to admit a full complexity
classification. Therefore the dichotomy for all PMMSNPs is most likely out of
reach for now. For that reason, we restrict our attention to certain particular
types of families $\mathcal F$ and $\mathcal G$.

\subsection*{Monochromatic cliques}

Consider the naturally arising graph colouring problem of the following form:
given a graph $\mathbb G$, decide whether there is a $c$-colouring of $\mathbb
G$ that avoids monochromatic $k$-cliques. This problem is readily seen to be an
MMSNP, where the forbidden patterns are monochromatic $k$-cliques, coloured with
one of $c$ colours; we denote such a family with $\templ{c}{k}$. This problem is
\NP-hard when $c, k \geq 2$ and $\max\{c,k\} > 2$, along the lines of
\Cref{ex:no_mono_tri}; and it is in P otherwise.

Such problems also exist in the graph theory under the name of \emph{$\mathbb
K_k$-free colouring}~\cite{chudnovsky}. The \emph{$\mathbb K_k$-free chromatic
number} $\chi_k(\mathbb G)$ of a graph $\mathbb G$ is the minimum number $c$
such that $\mathbb G$ has a $c$-colouring avoiding a monochromatic $k$-clique.
The computational complexity of $\mathbb K_k$-free colouring has already been
studied; for example, in~\cite{karpinski}, the authors prove that
$\MMSNP(\templ{c}{3})$ is NP-complete, for every $c\geq 2$.

We note that the problem $\templ{c}{k}$ cannot be expressed as a finite-domain
CSP if $k > 2$. Indeed, for any finite graph $\mathbb H$, there exists a graph
with girth greater than $k$ and chromatic number greater than that of $\mathbb
H$ (by~\cite{erdos_sparse_incomparability}). Any $c$-colouring of this graph is
$\templ{c}{k}$-free, but it does not admit a homomorphism to $\mathbb H$.

On the other hand, it is very easy to construct the infinite-domain CSP
templates for such problems. Clearly, the class of finite $\mathcal
F_k^{(c)}$-free coloured graphs has the amalgamation property, so there is no
need to add any new relations, because the Fra\"{\i}ss\'e-limit of this class is
already a finitely bounded homogeneous structure. The graph $\mathbb H_k^{(c)}$
such that $\CSP(\mathbb H_k^{(c)}) = \MMSNP(\templ{c}{k})$ is obtained from
$c$-many disjoint copies of the universal $\mathbb K_k$-free graph by connecting
every two vertices from distinct copies with an edge. So, the problems that
forbid monochromatic cliques are among the tamest MMSNP problems that are not
finite-domain CSPs.

In the second part of the paper, we investigate the complexity of the promise
version of these graph colouring problems, that is, Promise MMSNPs over a pair
of families $\templ{c}{k}$ and $\templ{d}{\ell}$. In this case, the following
two types of relations will often appear in the corresponding PCSP templates
obtained from \cref{th:mainequivalence}: first, for integers $0 \le k \le \ell$,
the $\ell$-ary relation \kinl{k}{\ell} consists of all tuples $\tuple a \in
\{0,1\}^\ell$ for which $\sum_i a_i = k$. Also, the \emph{Not All Equal}
relation $\NAE[d]{r} \subseteq [d]^r$ consists of all non-constant
tuples;\footnote{We sometimes omit the superscript $(d)$ if the domain is clear
from the context.} a tuple $\tuple a$ is constant if $a_i = a_j$ for all $i, j$.

The following example illustrates that PMMSNPs over forbidden monochromatic
cliques include non-trivial interesting problems in which the promise turns an
otherwise NP-complete MMSNP problem into an efficiently solvable one.

\begin{example}\label{ex:AIP}
    Let us consider the PMMSNP problem with template $(\mathcal{F}_3^{(2)},
    \mathcal{F}_4^{(2)})$. In this problem, we are given a graph $G = (V, E)$
    and must decide whether it has a 2-colouring that avoids monochromatic
    triangles, or not even a 2-colouring that avoids monochromatic 4-cliques. We
    will show that something slightly
    stricter is possible: if we assume the
    input graph has a 2-colouring avoiding all monochromatic triangles, we will
    find a 2-colouring avoiding all monochromatic 4-cliques. (Thus, if our
    algorithm successfully outputs a 2-colouring avoiding all monochromatic
    4-cliques we may answer \yes{}, otherwise we must answer \no{}.)
    
    To do this, construct a 4-uniform hypergraph whose vertices coincide with
    the vertices of $G$, and whose hyperedges are given by the 4-cliques of $G$.
    Call this hypergraph $H = (V, E')$.\footnote{Note that triangles in $G$ that
    belong to no 4-clique do not affect $H$.} The condition that $G$ has a
    2-colouring avoiding all monochromatic triangles implies that there exists a
    2-colouring $c : V \to \{0, 1\}$ of $H$ that, for every $(x, y, z, t) \in
    E'$, colours half of $x, y, z, t$ with $0$ and the other half with $1$. We
    must output a colouring $d : V \to \{0, 1\}$ that leaves no edge of $H$
    monochromatic. This is just the problem $\PCSP(\kinl{2}{4}, \NAE{4})$,
    well-known from the PCSP literature, which admits a poly-time solution via
    the \emph{Affine Integer Programming (AIP) relaxation}~\cite{BG21_pcsps}.
    Let us quickly present this algorithm.

    By assumption, there exists $c : V \to \{0, 1 \}$ such that
    \begin{equation}\label{eq:relaxation}
    \forall (x, y, z, t) \in E' : \qquad c(x) + c(y) + c(z) + c(t) = 2.
    \end{equation}
    Of course, it is \NP-hard to find this mapping; however, it is possible to
    solve the system of equations in~\eqref{eq:relaxation} over $\mathbb{Z}$, by
    computing the Smith normal form~\cite{smith_normal_form}; suppose
    $\tilde{c}$ is our solution. Now, consider any edge $(x, y, z, t)$. It is
    not possible that all of $\tilde{c}(x), \tilde{c}(y), \tilde{c}(z),
    \tilde{c}(t)$ are $\leq 0$ (or else the sum would be $\leq 0 < 2$); nor is
    it possible that all are $> 0$ (or else the sum would be at least $4 > 2$).
    Hence, if we set $v \in V$ to $1$ or $0$ depending on whether $\tilde{c}(v)
    \geq 1$ or $\tilde{c}(v) \leq 0$, we can find the desired solution in
    $\NAE{4}$.
\end{example}

\subsection*{Dictatorship tests and Rich 2-to-1 Conjecture}

As explained before, we reduce questions about Promise MMSNPs to questions about finite-domain Promise CSPs.
For the hardness part of our dichotomy, we do not follow the algebraic approach to PCSP \cite{Bible}, instead employing \emph{dictatorship tests}.

Dictatorship testing emerged as a common technique in the field of hardness of approximation in the 2000s, growing out of a broader PCP (Probabilistically Checkable Proofs) framework.
The key idea traces back to the introduction of the ``Long Code'' \cite{bgs98longcode}, a highly redundant encoding scheme which makes Fourier-analytic tools useful in the analysis of reductions from the Label Cover problem.
A dictatorship test examines a multivariate function, and aims to distinguish ``dictator'' functions -- those depending on exactly one coordinate -- from functions that are far from any dictator, in a certain sense.
This was used extensively to compose PCP-style reductions: if a solver could satisfy a target problem instance built from many correlated dictatorship tests, one could select a small set of ``influential'' coordinates that behaved like a consistent labelling of the original hard problem. This approach led to new inapproximability results for problems like Max-Cut \cite{khot2007maxcut}, Vertex Cover \cite{khot2008vc}, and graph colouring \cite{dinur2006approxcol}.
The technique reached its most influential form after Khot's introduction of the Unique Games Conjecture (UGC) in 2002 \cite{khot2002ugc}.

In his seminal work, Raghavendra \cite{raghavendra} discovered a pipeline that converts any dictatorship test with imperfect completeness into a hardness of approximately solving unsatisfiable CSP instances, assuming the Unique Games Conjecture.\footnote{Recently, OpenAI has claimed a proof of the Unique Games Conjecture~\cite{ugcopenai}.}
Braverman, Khot, Lifschitz, and Minzer \cite{multislices} extended this pipeline to satisfiable instances of \emph{connected} CSPs via perfect completeness dictatorship tests at the cost of Rich 2-to-1 Conjecture in place of UGC.
We extend this further to qualitative approximation, that is, Promise CSP; it is worth noting that the extension of the dictatorship testing framework to PCSPs is rather straightforward --- the main challenge is devising dictatorship tests for Promise CSP templates of interest.

Let us briefly introduce the Rich 2-to-1 Conjecture now (for the formal
definition, see \cref{conj:rich2to1}). The (perhaps) more famous Unique Games
Conjecture of Khot~\cite{khot2002ugc} states that it is \NP-hard to even
approximately solve an \emph{almost satisfiable} instance of Unique Games
(equivalently, a CSP instance in which each constraint is a bijection). This
conjecture has inspired numerous influential results on hardness of
inapproximability of various computational
problems~\cite{khot2008vc,khot2007maxcut}. On the other hand, it is trivial to
solve a \emph{satisfiable} instance of Unique Games, therefore this conjecture
has no use in tackling so called \emph{perfect completeness} approximation
problems, that is, problems of finding an approximate solution to
\emph{satisfiable} instances.
The other conjectures introduced by Khot in the same manuscript~\cite{khot2002ugc}, like the $d$-to-1 Conjecture, were not trivial in the perfect completeness regime, but they could not provide the reductions as general as UGC.\footnote{Recently, Fei, Minzer, and Wang~\cite{dto1proof} and OpenAI~\cite{2to1openai} have claimed proofs of the $d$-to-1 Conjecture with perfect completeness for $d=4$ and $d=2$, respectively.}
The Rich 2-to-1 Conjecture was proposed in 2021
by~\cite{rich2to1} to overcome both these issues. Namely, it states that it is \NP-hard to
approximately solve a \emph{satisfiable} instance of Rich 2-to-1 Games
(equivalently, a CSP instance in which each constraint is a 2-to-1 function with
an extra regularity property imposed). On \emph{almost satisfiable} instances,
the authors of~\cite{rich2to1} proved that this conjecture is equivalent to the
Unique Games Conjecture, which makes it a perfect-completeness surrogate of the
UGC.

The Promise MMSNPs over $\templ{c}{k}, \templ{d}{\ell}$ include the notorious
Approximate Graph Colouring (take $k=\ell=2$) which can be viewed as a PCSP over
two cliques. For example, $\PCSP(\mathbb K_3, \mathbb K_6)$ is equivalent to
$\PMMSNP(\templ{3}{2}, \templ{6}{2})$.
Assuming the $d$-to-1 Conjecture of~\cite{khot2002ugc}, \NP-hardness was proved for all $\PCSP(\mathbb
K_c, \mathbb K_d)$ such that $3 \le c \le d$ in~\cite{dinur2006approxcol,
GS2020agc}. Later,~\cite{multislices} suggested a substantially different
hardness proof under the recently proposed Rich 2-to-1
Conjecture~\cite{rich2to1}. Outside of the Approximate Graph Colouring, the Rich
2-to-1 Conjecture was the basis of several hardness results for Boolean
PCSPs~\cite{bgs23ordered,booleanfourier2026}. These (and other) conjectures
provide a way to bypass the limitations of the unconditional sources of
hardness, and make progress on various open inapproximability problems in order
to better understand their underlying mathematical structure, which could, in
turn, lead to new, stronger, unconditional sources of hardness in the future.

In a recent breakthrough, Fei, Minzer, and Wang~\cite{dto1proof} have claimed a proof of the 4-to-1 Conjecture, thus settling the computational complexity of the Approximate Graph Colouring.
Later, OpenAI have published preprints claiming proofs of the 2-to-1 Conjecture~\cite{2to1openai}, as well as the Unique Games Conjecture~\cite{ugcopenai}.
These results might be viewed as a step towards confirming the Rich 2-to-1 Conjecture.

\subsection*{Our contributions}

We view the contributions of this paper as twofold. Firstly, a classification
(resulting in a dichotomy) of all Promise MMSNP problems over forbidden
monochromatic cliques.

\begin{restatable}{theorem}{dichotomy}\label{th:dichotomy}
    Let $c, d, k, \ell$ be positive integers such that $\templ{c}{k} \to
    \templ{d}{\ell}$. Then
    \begin{itemize}
        \item $\ell \ge c(k-1)$ and $\PMMSNP(\templ{c}{k}, \templ{d}{\ell})$ is
            solvable in polynomial time, or 
        \item $\ell < c(k-1)$ and $\PMMSNP(\templ{c}{k}, \templ{d}{\ell})$ is
            \NP-hard, assuming the Rich 2-to-1 Conjecture.
    \end{itemize}
\end{restatable}

As the first step towards \cref{th:dichotomy}, we establish that every Promise
MMSNP of interest is poly-time equivalent to a certain finite-domain Promise CSP
in which both the promise and the target templates have a single, symmetric
relation. The PCSPs of this form are readily seen to be expressible as certain
promise hypergraph colouring problems. For concreteness, we deal with relations
over a universe $[c]$ that consist of all tuples of length $\ell$ in which each
$a \in [c]$ appears less than $k$ times. Such a relation is denoted with
$\urel{c}{k}{\ell}$. For example, $\urel{2}{3}{3}$ is the ternary \emph{Not All
Equal} relation. Consequently, we show that every promise MMSNP over two
monochromatic clique families is poly-time equivalent to a PCSP defined by a
pair of relations of this form.

\begin{restatable}[Special case of \cref{th:mainequivalence}]{lemma}{cliqueequivalence}\label{lemma:clique-equivalence}
    Let $c, d, k, \ell$ be positive integers. If $\templ{c}{k} \to
    \templ{d}{\ell}$, then the problems $\PMMSNP(\templ{c}{k}, \templ{d}{\ell})$
    and $\PCSP(\urel{c}{k}{\ell}, \urel{d}{\ell}{\ell})$ are poly-time
    equivalent.
\end{restatable}

Thus, the desired dichotomy (\cref{th:dichotomy}) amounts to classifying all
PCSPs of this form.

\begin{restatable}{theorem}{pcspdichotomy}\label{cor:dichotomy}
    Let $c, d, k, \ell$ be positive integers such that $\templ{c}{k} \to
    \templ{d}{\ell}$. Then
    \begin{itemize}
        \item $\ell \ge c(k-1)$ and $\PCSP(\urel{c}{k}{\ell}, \NAE[d]{\ell})$ is
            solvable in polynomial time via AIP, or
        \item $\ell < c(k-1)$ and $\PCSP(\urel{c}{k}{\ell}, \NAE[d]{\ell})$ is
            \NP-hard, assuming the Rich 2-to-1 Conjecture.
    \end{itemize}
\end{restatable}

In fact, we will demonstrate that the complexity of a $\PCSP(R, \NAE{\ell})$ is
governed by certain connectivity properties of the left-hand side relation. To
formalize this notion, let us call a relation \emph{reconfigurable} if any of
its members can be obtained from any other member by modifying one coordinate at
a time without leaving the relation at any point.

In particular, we prove the following hardness result for PCSPs, which is the
technical core and the second contribution of this paper. Let $\NAE[d]{r}$
denote the $r$-ary \emph{Not All Equal} relation on $d$ elements.

\begin{restatable}{theorem}{blackbox}\label{th:blackbox}
    Let $\emptyset \neq R \subseteq A^r$ be a symmetric\footnote{A relation is symmetric if it is closed under permutations of coordinates.} reconfigurable\footnote{In fact, both assumptions can be relaxed (see \cref{sec:hardness}).} relation.
    For any $d$ such that $R \to \NAE[d]{r}$, the problem $\PCSP(R, \NAE[d]{r})$ is \NP-hard, assuming the Rich 2-to-1 Conjecture.
\end{restatable}

It is an extension of the work of~\cite{multislices}: they proved this theorem
in the special case of $R = \{(a, b) \in \{1,2,3\} \mid a \neq b\}$, which
coincides with the Approximate Graph Colouring. It is worth noting that their
proof in fact covers \cref{th:blackbox} for $r=2$, although the authors do not
mention it explicitly. As a side note, a symmetric $R \subseteq A^2$ (where
every value in $A$ appears in $R$) is reconfigurable if and only if the
(undirected) graph $(A; R)$ is connected and not
bipartite~\cite{KOWZ23_topology}. Thus, the Brakensiek-Guruswami
conjecture~\cite{BG21_pcsps}, which postulates that $\PCSP(G, H)$ is \NP-hard
whenever $G$ is not bipartite and $H$ is loopless, already follows from the case
$r=2$ of \cref{th:blackbox}, assuming the Rich 2-to-1 Conjecture.

Our approach extends the pipeline of \cite{multislices} that converts any perfect-completeness dictatorship test into a hardness of approximately solving satisfiable instances of reconfigurable\footnote{In this and further statements, reconfigurability can be relaxed to a more technical condition. We define it in \cref{sec:hardness}.} CSPs, assuming the Rich 2-to-1 Conjecture, to Promise CSPs.
The definition of a dictatorship test is quite technical; we shall attempt to convey an informal intuition for PCSP templates with a single relational symbol here, while the full formal definition ensues in \cref{sec:hardness} and in the Appendix.
\begin{definition}[Informal]
    Let $((A; R^\mathbb A), (B; R^\mathbb B))$ be a PCSP template.
    A probability distribution $\Omega$ over $R^\mathbb A$ is called a \emph{dictatorship test} for $(\mathbb A, \mathbb B)$ if any
    function $f \colon A^n \to B$ that is ``far from any dictator'' applied row-wise to a matrix with each column sampled from $\Omega$ independently statistically does not always output an element of $R^\mathbb B$.
\end{definition}

The existence of a dictatorship test with a reconfigurable support facilitates a reduction from the Rich 2-to-1 Conjecture to the PCSP.
The proof of this fact is heavily analytical; in
particular, it relies on the \emph{invariance principle} on \emph{multislice}
domains developed in \cite{multislices}.

\begin{theorem}[Based on {\cite[Theorem 6.7]{multislices}}]
    Let $(\mathbb A, \mathbb B)$ be a PCSP template.
    If it has a dictatorship test $\Omega$ with a reconfigurable support, then $\PCSP(\mathbb A, \mathbb B)$ is \NP-hard, assuming the Rich 2-to-1 Conjecture.
\end{theorem}

We thus prove \cref{th:blackbox} by devising appropriate dictatorship tests.
\begin{lemma}
    Let $\emptyset \neq R \subseteq A^r$ be a symmetric reconfigurable relation.
    Any PCSP template $(R, \NAE[d]{r})$ has a dictatorship test with a reconfigurable support.
\end{lemma}

In the language of~\cite{austrin2025usefulness}, \cref{th:blackbox} states that
symmetric reconfigurable relations are \emph{promise-useless} --- the extra
information about the input instance that they provide does not help to find
even a $d$-colouring that leaves no edge monochromatic (and hence cannot help us
find a solution to \emph{any} nontrivial hypergraph colouring problem). This
result builds, albeit conditionally, on the hardness results
of~\cite{brokenpromise} (where they study PCSP templates of the form $(R,
\NAE{3})$ such that no tuple in $R$ is injective) and~\cite{beyond1in3nae}
(where they study PCSP templates of the form $(\kinl{k}{\ell} \cup S,
\NAE[2]{\ell})$ for some relation $S$ over $\{0,1\}$).

Besides, the notion of reconfigurability in \cref{th:blackbox} hints towards its
topological nature --- we elaborate on this question in the
\nameref{sec:conclusions}. It is interesting to note that this notion has
appeared before in the Fourier-analytical research as well,
e.g.~Mossel~\cite{mossel} considers a relation \emph{connected} if it is
reconfigurable.

\begin{remark}\label{rem:nice-condition}
    Our hardness result holds, in particular, for $\PCSP(R, S)$ whenever $R$ is \balanced\ and reconfigurable, and $S$ has no constant tuple. In fact, this PCSP remains NP-hard even if we enlarge $R$. We note an interesting fact motivated by this: A relation $R$ contains some \balanced\ and reconfigurable $Q \subseteq R$ if and only if there exists a matrix $M$ whose columns belong to $R$, whose rows are permutations of each other, and where for every column $c$ (other than the first one) there exists a column $c'$ to its left which differs from it in at most one position.

    In particular, this implies a host of new \emph{minor conditions} sufficient for hardness under the Rich 2-to-1 Conjecture; see \Cref{app:newMinorConditions} for details.
\end{remark}

To illustrate the power of \cref{th:blackbox}, we give an example from the world
of approximate hypergraph colouring. Given a hypergraph $H = (V, E)$, a
\emph{Linearly Ordered (LO)} $c$-colouring is an assignment $V \to [c]$ to the
vertices of $H$ such that, in every hyperedge of $H$, its maximum colour appears
exactly once. Quite some work has been done on the hardness of approximate LO
colouring of uniform hypergraphs~\cite{bbb21_los, nz23_los,nvwz25_los,kv25_los};
for now, let us focus on the case of 3-uniform hypergraphs. For the following,
let $(3, c, d)$-approximate LO colouring denote the problem of distinguishing
whether a 3-uniform hypergraph is LO $c$-colourable, or not even LO
$d$-colourable. The state-of-the-art hardness result is the following.

\begin{theorem}[Simple reduction from~\cite{kv25_los}]
    $(3, c, c + 1)$-approximate LO colouring is \NP-hard for $c \geq 2$.
\end{theorem}

\Cref{th:blackbox} completes the picture for $c \geq 3$, under the Rich 2-to-1
conjecture. Since the template corresponding to LO $c$-colouring a 3-uniform
hypergraph with $c \geq 3$ colours is reconfigurable (cf.~\cref{fig:lo3}), we
deduce:

\begin{corollary}\label{corr:los}
    $(3, c, d)$-approximate LO colouring is \NP-hard for all $3 \leq c \leq d$, assuming the Rich 2-to-1 Conjecture. 
\end{corollary}

This confirms (except for the case $c = 2$) the conjecture of Barto, Battistelli
and Berg~\cite{bbb21_los} on approximate LO colouring of 3-uniform hypergraphs,
under the Rich 2-to-1 conjecture. We also remark that there is an essential link
between a topological proof of the hardness of $(3, 3, 4)$-approximate LO
colouring~\cite{FNOTW25_topology} and our proof here: the reconfigurability is
an essential implication of~\cite{FNOTW25_topology}.

\begin{figure}
    \centering
    \begin{tikzpicture}
\node (001) at (0:4cm) {112};
\node (010) at (120:4cm) {121};
\node (100) at (240:4cm) {211};
\node (002) at ($(001)!1/3!(010)$) {113};
\node (020) at ($(010)!1/3!(100)$) {131};
\node (200) at ($(001)!2/3!(100)$) {311};

\node (012) at ($(002)!1/2!(010)$) {123};
\node (102) at ($(002)!1/4!(100)$) {213};
\node (021) at ($(020)!1/4!(001)$) {132};
\node (120) at ($(100)!1/2!(020)$) {231};
\node (201) at ($(200)!1/2!(001)$) {312};
\node (210) at ($(200)!1/4!(010)$) {321};

\node (211) at ($(201)!1/2!(210)$) {322};
\node (121) at ($(021)!1/2!(120)$) {232};
\node (112) at ($(012)!1/2!(102)$) {223};

\draw (001) -- (002);
\draw (001) -- (021);
\draw (001) -- (201);
\draw (002) -- (012);
\draw (002) -- (102);
\draw (010) -- (012);
\draw (010) -- (020);
\draw (010) -- (210);
\draw (012) -- (112);
\draw (020) -- (021);
\draw (020) -- (120);
\draw (021) -- (121);
\draw (100) -- (102);
\draw (100) -- (120);
\draw (100) -- (200);
\draw (102) -- (112);
\draw (120) -- (121);
\draw (200) -- (201);
\draw (200) -- (210);
\draw (201) -- (211);
\draw (210) -- (211);
    \end{tikzpicture}
    \caption{Reconfiguration graph for $\LO_3$. The tuple  $(a, b, c)$ will be rendered $abc$ for compactness.}\label{fig:lo3}
\end{figure}
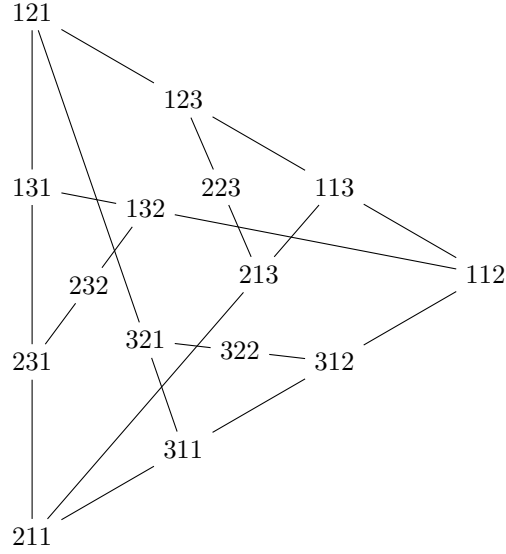

It is worth noting that $(4, c, d)$-approximate LO colouring was recently shown
to be unconditionally \NP-hard for all $3 \le c \le d$ by~\cite{nvwz25_los}.

The more general version of \cref{th:blackbox} with relaxed assumptions on $R$, which we state in \cref{sec:hardness}, also answers an open question by \cite{austrin2025usefulness}: namely, what is the complexity of $\PCSP(R, \NAE[2]{r})$ where $R \subseteq \{0,1\}^5$ is obtained by closing the tuple $(1,1,0,0,0)$ under cyclic shifts and flipping the roles of $0$ and $1$?
Our hardness result implies hardness under the Rich 2-to-1 Conjecture. In particular, this even follows from \Cref{rem:nice-condition}.

\subsection*{Related work}

We shall mention some previous work in the context of \cref{th:blackbox}: it
consists in various rainbow/balanced hypergraph colouring inapproximability
results; we will use the PCSP parlance for clarity.

Observe that $\PCSP(\urel{c}{\ell}{\ell}, \urel{d}{\ell}{\ell})$ is just the
approximate non-monochromatic hypergraph colouring problem. This problem serves
as one of the key sources of hardness for PCSPs in general: see~\cite{Bible},
which proves that a PCSP is \NP-hard by a gadget reduction from this problem if
and only if it fails to admit an \emph{Ol\v{s}\'{a}k polymorphism}, a condition
which can be checked algorithmically; through this condition they were able to
prove hardness of the approximate graph colouring $\PCSP(\mathbb K_c, \mathbb
K_{2c-1})$ for $c \ge 3$.

In \cref{table:results}, we summarize existing hardness results for several
regimes of promise hypergraph colouring that are subsumed, albeit conditionally,
by this work. Finally, PCSPs of our form were mentioned
in~\cite{sandeep2022packing} as \emph{Balanced Hypergraph Colouring} problems.

\renewcommand{\arraystretch}{1.5}
\begin{table}[ht]
    \centering
    \caption{A summary of hypergraph colouring hardness results. These results only hold for choices of parameters where the PCSP template is valid, i.e.~only when $\mathbb{A} \to \mathbb{B}$ for $\PCSP(\mathbb{A}, \mathbb{B})$.}\label{table:results}
    \begin{tabular}{|c|c|c|c|}
    \hline 
    Problem & Parameters & Assumptions & Reference  \\ 
    \hline
    $\PCSP(\urel{c}{\ell}{\ell}, \urel{d}{\ell}{\ell})$ & $\ell \ge 3$ and $d \ge c \ge 2$ & --- & \cite{DRS05,nvwz25_los} \\ 
    \hline
    $\PCSP(\urel{\ell+1}{3}{\ell}, \urel{d}{\ell}{\ell})$ & $\ell,d \ge 3$ & V Label Cover\tablefootnote{This conjecture, which is, in some sense, a hypergraph generalisation of the ``fish-shaped'' conjecture of Dinur et al. \cite{dinur2006approxcol}, was proposed in \cite{BG2021quest}.} & \cite{BG2021quest}  \\
    \hline
    $\PCSP(\urel{c}{q+c}{cq}, \urel{d}{cq}{cq})$ & $q,c,d \ge 2$ & --- & \cite{GL2018balancedrainbow} \\
    \hline
    $\PCSP(\urel{c}{q+d}{cq}, \urel{d}{cq}{cq})$ & $q,c,d \ge 2$ & --- & \cite{GS2018rainbow}\tablefootnote{They prove it only for even $c$ and $d=2$, but with slight modifications to their proof, it is possible to generalise it to any $c,d \ge 2$.} \\
    \hline
    $\PCSP(\urel{c}{k}{\ell}, \urel{d}{\ell}{\ell})$ & $2 \leq \ell < c(k - 1)$ & Rich 2-to-1 & This work \\\hline
    \end{tabular}
\end{table}

\section{Preliminaries}
\subsection*{Notation}
We write $[n]$ for the set $\{1, \dots, n\}$ and $\mathbb N$ for the set of
positive integers. Given a set $A$, any element of $A^n$ is an \emph{$n$-ary
tuple}. A tuple $\tuple x \in A^n$ can be treated as a function $\tuple x\colon
[n] \to A$. Conversely, we sometimes treat any function from a finite set to $A$
as a tuple, implicitly enumerating the elements of this finite set. The $i$-th
entry of a tuple $\tuple x \in A^n$ is denoted by $x_i$.

A (constraint) \emph{language} is a set $\sigma$ whose every member
$R\in\sigma$, called a \emph{relation symbol}, is assigned a positive integer
$\ar(R)$, called the \emph{arity} of $R$. For $\sigma=\{R_1,\ldots,R_n\}$, a
(relational) \emph{$\sigma$-structure} is a tuple $\mathbb A = (A, R_1^\mathbb
A,\ldots, R_k^\mathbb A)$, where $A$ is a set and $R_i^\mathbb A \subseteq
A^{\ar(R_i)}$, for each $i\in[k]$. The set $A$ is called the \emph{domain} of
$\mathbb A$ and, for every $i \in[k]$, $R_i^\mathbb A$ is called the
\emph{interpretation of $R_i$ in $\mathbb A$} or simply a \emph{relation}. In
this paper, all structures are to be assumed finite, unless stated otherwise.

For two $\sigma$-structures $\mathbb A$ and $\mathbb B$, a \emph{homomorphism}
from $\mathbb A$ to $\mathbb B$ is a map $h\colon\mathbb A \to \mathbb B$ that
preserves relations.

\subsection*{MMSNP}

Let $\tau$ be a fixed language and $\mathcal F$ be a finite family of pairs
$(\mathbb F, f)$, where $\mathbb F$ is a finite $\tau$-structure, and $f\colon
F\to [c]$ for some $c\in\mathbb N$. The set $[c]$ is called the \emph{colour
set} of $\mathcal F$ and each $(\mathbb F,f)\in\mathcal F$ is called a
\emph{forbidden pattern} of $\mathcal F$.

A \emph{$c$-expansion} (or a \emph{$c$-colouring}) of a $\tau$-structure
$\mathbb I$ is a pair $(\mathbb I, g)$ where $g\colon I \to [c]$. If the number
of colours is not important, we will just write an \emph{expansion} (or a
\emph{colouring}). Conversely, the \emph{reduct} of $(\mathbb I, g)$ is the
$\tau$-structure $\mathbb I$.

We say that an expansion $(\mathbb I, g)$ is \emph{$\mathcal F$-free} if every
$(\mathbb F, f)\in\mathcal F$ and every homomorphism $h$ from $\mathbb F$ to
$\mathbb I$ satisfies $f\neq g\circ h$.

\begin{definition}[MMSNP]
    Let $\tau$ be a constraint language. For a finite family $\mathcal F$ of
    forbidden patterns over $\tau$, the problem $\MMSNP(\mathcal F)$ comes in
    two flavours:
    \begin{description}
        \item[(decision)] given a $\tau$-structure $\mathbb I$, output \yes{} if
            $\mathbb I$ has an $\mathcal F$-free expansion, and output \no{}
            otherwise;
        \item[(search)] given a $\tau$-structure $\mathbb I$, output an
            $\mathcal F$-free expansion.
    \end{description}
\end{definition}
The decision MMSNP is poly-time equivalent to the search MMSNP: it follows from
the poly-time equivalence of MMSNP and precoloured MMSNP~\cite[Corollary
6.10]{bodirsky_madelaine_mottet}. For the sake of brevity, we will write
$\mathcal F$ instead of $\MMSNP(\mathcal F)$ to denote the decision problem
associated with $\mathcal F$.

Let $(\mathbb I, g)$ be a $c$-expansion of $\mathbb I$. Given a mapping $r\colon
[c] \to [d]$, any $d$-expansion $(\mathbb I, h)$ that satisfies $r \circ g = h$
is called an \emph{$r$-preimage} of $(\mathbb I, g)$.

Let $\mathcal F$ and $\mathcal G$ be two families of forbidden patterns over the
same constraint language with colour sets $[c]$ and $[d]$, respectively. A
mapping $r\colon[c]\to[d]$ is a \emph{recolouring} from $\mathcal F$ to
$\mathcal G$, if no forbidden pattern of $\mathcal G$ has a $\mathcal F$-free
$r$-preimage. We write $\mathcal F\to \mathcal G$ if there is a recolouring from
$\mathcal F$ to $\mathcal G$.

\begin{definition}[PMMSNP]
    Let $\tau$ be a constraint language. For finite families of forbidden
    patterns $\mathcal F$ and $\mathcal G$ over $\tau$ such that $\mathcal F\to
    \mathcal G$, the problem $\PMMSNP(\mathcal F, \mathcal G)$ comes in two
    flavours:
    \begin{description}
        \item[(decision)] given a $\tau$-structure $\mathbb I$, output \yes{} if
            $\mathbb I$ has an $\mathcal F$-free expansion, and output \no{} if
            $\mathbb I$ even does not have a $\mathcal G$-free expansion;
        \item[(search)] given a $\tau$-structure $\mathbb I$ promised to have an
            $\mathcal F$-free expansion, output a $\mathcal G$-free expansion of
            $\mathbb I$.
    \end{description}
\end{definition}

\subsection*{CSP}

It was already mentioned that every problem in MMSNP is poly-time equivalent to
a \emph{Constraint Satisfaction Problem} (CSP)~\cite{federvardi1998}.

\begin{definition}
    Let $\sigma$ be a constraint language. For a $\sigma$-structure $\mathbb B$,
    the problem $\CSP(\mathbb B)$ comes in two flavours:
    \begin{description}
        \item[(decision)] given a $\sigma$-structure $\mathbb I$, output \yes{}
            if $\mathbb I \to \mathbb B$, and output \no{} if $\mathbb I
            \not\to\mathbb B$;
        \item[(search)] given a $\sigma$-structure $\mathbb I$, output a
            homomorphism from $\mathbb I$ to $\mathbb B$.
    \end{description}
\end{definition}
\noindent The decision CSP is poly-time equivalent to the search
CSP~\cite[Corollary~4.9]{BJK05}.\@

\begin{definition}[PCSP]
    Let $\sigma$ be a relational language. For any two $\sigma$-structures
    $\mathbb A$ and $\mathbb B$ such that $\mathbb A \to \mathbb B$, the problem
    $\PCSP(\mathbb A, \mathbb B)$ comes in two flavours:
    \begin{description}
        \item[(decision)] given a $\sigma$-structure $\mathbb I$, output \yes{}
            if $\mathbb I \to \mathbb A$, and output \no{} if $\mathbb I
            \not\to\mathbb B$;
        \item[(search)] given a $\sigma$-structure $\mathbb I$ promised to admit
            a homomorphism to $\mathbb A$, output a homomorphism from $\mathbb
            I$ to $\mathbb B$.
    \end{description}
\end{definition}
\noindent The search PCSP is at least as hard as the decision PCSP, however the
converse is not known. We provide algorithms for the search version, and proofs
of hardness for the decision version.

``Sandwiching'' is a trivial reduction technique between PCSPs:

\begin{observation}[See~\cite{BG21_pcsps,Bible}]\label{obs:sandwich}
    If $\mathbb A \to \mathbb B \to \mathbb C \to \mathbb D$, then
    $\PCSP(\mathbb A, \mathbb D)$ reduces to $\PCSP(\mathbb B, \mathbb C)$ in
    poly-time.
\end{observation}

In the case $\mathbb B = \mathbb C$, the problem $\PCSP(\mathbb B, \mathbb C)$
is just a CSP, namely $\CSP(\mathbb B) = \CSP(\mathbb C)$.

\sloppy
For $R \subseteq A^r$ and $S \subseteq B^r$, we sometimes abuse the notation and
write $\PCSP(R, S)$ in place of $\PCSP((A; R), (B; S))$ for convenience.

\subsection*{Rich 2-to-1 Conjecture}

We will have to concede the CSP formulation of the Rich 2-to-1 Conjecture, used
in the \nameref{sec:intro}, in favour of the more commonly used \emph{Label
Cover} one. For the following, a map $\pi : [2n] \to [n]$ is called
\emph{2-to-1} if $|\pi^{-1}(i)| = 2$ for each $i \in [n]$.

\begin{definition}[Rich 2-to-1 Label Cover]\label{defn:2to1}
    A Rich 2-to-1 Label Cover instance is a tuple $(U \cup V, E, [2n], [n], \{
    \pi_e \}_{e \in E})$, where $(U \cup V, E)$ form a bi-regular bipartite
    graph, and each $\pi_e : [2n] \to [n]$ is a 2-to-1 map. Furthermore, for
    every fixed $u \in U$, if we sample $v$ from the neighbourhood of $u$
    uniformly at random, then the function $\pi_{uv}$ is uniformly distributed
    among the set of 2-to-1 maps $[2n] \to [n]$ (this property is called
    \emph{richness}).\footnote{In other words, if one fixes any $u \in U$ and
    considers the maps $\pi_e$ associated with the edges incident to $u$, every
    2-to-1 map appears equally often.}

    A \emph{labelling} for this instance is an assignment $c$ which sends $U$ to
    $[2n]$ and $V$ to $[n]$. The value of the labelling is the probability that
    $c(v) = \pi_{uv}(c(u))$ for $(u, v)$ sampled uniformly at random from $E$.
    The value of the instance is the maximum value of any labelling.
\end{definition}

\begin{conjecture}[Rich 2-to-1 conjecture~\cite{rich2to1}]\label{conj:rich2to1}
    For every $\varepsilon > 0$ there exists a large enough $n$ such that it is
    \NP-hard to decide whether a Rich 2-to-1 Label Cover instance $(U \cup V, E,
    [2n], [n], \{ \pi_e\}_{e \in E})$ has value 1, or value $< \varepsilon$.
\end{conjecture}

\section{Equivalence with PCSP}\label{sec:equivalence}
In this section, we show that the problem $\PMMSNP(\mathcal F,\mathcal G)$ is
poly-time equivalent to some finite-domain $\PCSP$. The reductions generalise
the ones of Feder and Vardi~\cite{federvardi1998}, where the finite-domain CSP
encoded all valid colourings of the reducts of the forbidden structures.
However, in the promise setting, the two MMSNP sentences could have different
sets of reducts, so applying Feder and Vardi's reduction to each MMSNP sentence
would give two finite structures over different constraint languages. We resolve
this issue.

\mainequivalence*

We first argue that $\mathcal F$ and $\mathcal G$ can be assumed to be in
certain \emph{normal form}. A structure $\mathbb F$ is \emph{biconnected} if for
every $x\in F$, the substructure of $\mathbb F$ induced on $F\setminus\{x\}$ is
not isomorphic to the disjoint union of two structures. A family $\mathcal F$ is
in \emph{normal form} if every $\mathbb F\in \mathcal F$ is biconnected and if
$\mathcal F$ is closed under taking quotients, i.e., identifying any two
elements in any $\mathbb F\in \mathcal F$ results in a structure from $\mathcal
F$. The following result allows us to assume the normal form without loss of
generality.

\begin{lemma}[Proposition 3.3 and Lemma 4.4 in~\cite{bodirsky_madelaine_mottet}]\label{lem:normal_form}
    Every MMSNP problem is a finite disjunction of MMSNP problems in normal form.
\end{lemma}

Further in this section, let $\mathcal F = \mathcal F_1\vee\dots\vee\mathcal
F_m$ and $\mathcal G=\mathcal G_1\vee\dots\vee\mathcal G_n$ be two disjunctions
of MMSNPs in normal form such that $\mathcal F\to \mathcal G$. Denote the colour
set of each $\mathcal F_i$ (resp. $\mathcal G_j$) by $[c_i]$ (resp. $[d_j]$) for
some $c_i \in \mathbb N$ (resp. $d_j \in \mathbb N$).

\begin{observation}\label{obs:containment}
    For every $i\in[m]$ there is $j\in[n]$ such that $\mathcal F_i\to \mathcal
    G_j$.
\end{observation}

We shall define $\mathbb S$ and $\mathbb T$ via a generalisation of the
construction given in the \nameref{sec:intro}. Let $\tau$ be the language of
both $\mathcal F$ and $\mathcal G$. For each $j\in[n]$, let $\mathcal G^*_j$
denote the set of all reducts of forbidden patterns in $\mathcal G_j$; for each
$i\in[m]$, define $\mathcal F_i^*$ similarly. Denote by $\sigma$ the language
which contains, for every $j \in [n]$ and every $\mathbb G \in \mathcal G^*_j$,
a $|G|$-ary symbol $R_\mathbb G$. We also expand $\sigma$ with an auxiliary
binary symbol $\sim$. This will be the constraint language of both $\mathbb S$
and $\mathbb T$.

Fix $j \in [n]$; let $\mathbb T_j$ be the $\sigma$-structure with domain
$[d_j]$, where

\begin{itemize}
    \item for every $\mathbb G\in \mathcal G_j^*$, the relation $R_\mathbb
        G^{\mathbb T_j} \coloneqq \left\{ \bar t\colon G\to [d_j] \mid (\mathbb
        G, \bar t)\not\in\mathcal G_j \right\}$;
    \item for every $j'\in[n]\setminus\{j\}$ and every $\mathbb G\in \mathcal
        G_{j'}^*$, the relations $R_\mathbb G^{\mathbb T_j} \coloneqq [d_j]^G$
        and $\sim^{\mathbb T_j}\coloneqq [d_j]^2$ are full.
\end{itemize}

Fix $i \in [m]$; by \cref{obs:containment}, $\mathcal F_i \to \mathcal G_j$ for
some $j \in [n]$. Let $\mathbb S_i$ be the $\sigma$-structure with domain
$[c_i]$, where

\begin{itemize}
    \item for every $\mathbb G\in \mathcal G_j^*$, the relation $R_\mathbb
        G^{\mathbb S_i} \coloneqq \{\bar s\colon G\to [c_i] \mid (\mathbb G,
        \bar s) \text{ is $\mathcal F_i$-free}\}$;
    \item for every $j'\in[n]\setminus\{j\}$ and every $\mathbb G\in \mathcal
        G_{j'}^*$, the relations $R_\mathbb G^{\mathbb S_i} \coloneqq [c_i]^G$
        and $\sim^{\mathbb S_i}\coloneqq [c_i]^2$ are full.
\end{itemize}
Finally, we define $\mathbb S\coloneqq\mathbb S_1\sqcup\dots\sqcup\mathbb S_m$
and $\mathbb T\coloneqq\mathbb T_1\sqcup\dots\sqcup\mathbb T_n$, where $\sqcup$
stands for the disjoint union. We will argue that $\PMMSNP(\mathcal F, \mathcal
G)$ is poly-time equivalent to $\PCSP(\mathbb S, \mathbb T)$, concluding the
proof of \cref{th:mainequivalence}.

\begin{claim}\label{claim:sigma}
    $\PMMSNP(\mathcal F,\mathcal G)$ poly-time reduces to $\PCSP(\mathbb
    S,\mathbb T)$.
\end{claim}

\begin{proof}
    Given an instance of $\PMMSNP(\mathcal F, \mathcal G)$ --- a
    $\tau$-structure $\mathbb X$ --- we construct an instance of $\PCSP(\mathbb
    S, \mathbb T)$ --- a $\sigma$-structure which we denote by $\sigma(\mathbb
    X)$. Namely, let $\sigma(\mathbb X)$ be the structure with the same domain
    as $\mathbb X$ in which every relation $R_\mathbb G^{\sigma(\mathbb X)}
    \coloneqq \Hom(\mathbb G, \mathbb X)$ is defined as the set of all
    homomorphisms from $\mathbb G$ to $\mathbb X$. The relation
    $\sim^{\sigma(\mathbb X)} \coloneqq X^2$ is full.
    
    If $\mathbb X$ is a \yes{} instance of $\MMSNP(\mathcal F)$, then there is
    $s\colon X\to [c_i]$ for some $i \in [m]$ such that $(\mathbb X, s)$ is
    $\mathcal F_i$-free. It follows from the construction of $\sigma(\mathbb X)$
    and $\mathbb S_i$ that $s$ is a homomorphism from $\sigma(\mathbb X)$ to
    $\mathbb S_i$.

    Conversely, if $h\colon X\to T$ is a homomorphism from $\sigma(\mathbb X)$
    to $\mathbb T$, then the image of $h$ is contained in some $\mathbb T_j$, as
    $\sim^{\sigma(\mathbb X)}=X^2$. It follows from the construction of
    $\sigma(\mathbb X)$ that $h$ colours every copy of every $\mathbb
    G\in\mathcal G_j$ in an allowed way, so $(\mathbb X, h)$ is $\mathcal
    G$-free.
\end{proof}

For the converse reduction, we make use of the Sparse Incomparability Lemma
which allows us to assume no short cycles in structures on input to
$\PCSP(\mathbb S, \mathbb T)$. We say that $\mathbb S$ has girth $>k$ if for any
choice of $\ell\leq k$ relational tuples $\bar s_1,\ldots,\bar s_\ell$ of
arities $r_1,\ldots,r_\ell$ in $\mathbb S$, the total number of elements
appearing in the tuples is greater than $\sum_{i=1}^\ell (r_i-1)$. That is, $k$
or fewer relational tuples never form a cycle.

\begin{lemma}[Sparse Incomparability Lemma~\cite{federvardi1998, kun2013}]\label{lem:kun_girth}
    Let $\sigma$ be a finite relational language, and $k,t \in \mathbb N$. For
    every $\sigma$-structure $\mathbb I$ there exists a poly-time constructible
    $\sigma$-structure $\mathbb I'$ with girth $>k$ such that for every
    $\sigma$-structure $\mathbb T$ of size $<t$ the equivalence $\mathbb I \to
    \mathbb T \Leftrightarrow \mathbb I' \to \mathbb T$ holds. Moreover,
    $\mathbb I'\to \mathbb I$.
\end{lemma}

Furthermore, we can assume that input structures are \emph{connected}, i.e.~not
isomorphic to a disjoint union of two structures. To see why this is the case,
observe that $\mathbb A \sqcup \mathbb B \to \mathbb S$ if and only if $\mathbb
A \to \mathbb S$ and $\mathbb B \to \mathbb S$. This implies that we can solve a
disconnected instance of $\PCSP(\mathbb S, \mathbb T)$ by solving its every
connected component. This means that $\mathbb I\to \mathbb S$ if and only if
$\mathbb I\to \mathbb S_i$, for some $i\in[m]$, and similarly $\mathbb
I\to\mathbb T\Leftrightarrow \mathbb I\to \mathbb T_j$, for some $j\in[n]$.
Since $\sim^{\mathbb T_j} = T_j^2$ is full, we can also assume that
$\sim^\mathbb I = I^2$ is full.

\begin{claim}\label{claim:tau}
    $\PCSP(\mathbb S,\mathbb T)$ poly-time reduces to $\PMMSNP(\mathcal
    F,\mathcal G)$.
\end{claim}
\begin{proof}
    Given an instance of $\PCSP(\mathbb S, \mathbb T)$ --- a $\sigma$-structure
    $\mathbb I$ --- we construct an instance of $\PMMSNP(\mathcal F, \mathcal
    G)$ --- a $\tau$-structure which we denote by $\tau(\mathbb I)$.
    
    We first obtain $\mathbb I'$ by applying \cref{lem:kun_girth} with
    parameters $k>\max\{|G| \mid \mathbb G\in \mathcal F_1 \cup\dots\cup \mathcal F_m\cup \mathcal G_1\cup\dots\cup  \mathcal G_n\}$ and $t >
    \max\{|S|, |T|\}$ and also taking the auxiliary relation $\sim^{\mathbb I'}
    = I'^2$ to be full. Denote by $\tau(\mathbb I)$ the $\tau$-structure with
    the same domain as $\mathbb I'$ which is obtained by replacing every $\bar t
    \in R_\mathbb G^{\mathbb I'}$ with a copy of $\mathbb G$ induced on $\bar
    t$, for every $\mathbb G \in \mathcal G_j^*$. Since all $\mathcal F_i$ and
    $\mathcal G_j$ are in normal form and since $\mathbb I'$ has large girth,
    every pair $\mathcal F_i \to \mathcal G_j$ has the following property: if
    there is a colouring $s$ of $\tau(\mathbb I)$ that leaves every copy of
    $\mathbb G$ $\mathcal F_i$-free (resp. $\mathcal G_j$-free) for each
    $\mathbb G\in \mathcal G_j^*$, then $(\tau(\mathbb I), s)$ is $\mathcal
    F_i$-free (resp. $\mathcal G_j$-free).
    
    If $\mathbb I \to \mathbb S_i$, then also $\mathbb I'\to \mathbb I \to
    \mathbb S_i$ by \cref{lem:kun_girth}. If $h$ is a homomorphism from $\mathbb
    I'$ to $\mathbb S_i$, then it colours all the copies of structures from
    $\mathcal G_j^*$ in $\tau(\mathbb I)$ such that each of them is $\mathcal
    F_i$-free, where $j$ is witnessing the containment $\mathcal F_i \to
    \mathcal G_j$. This implies that $(\tau(\mathbb I), h)$ is $\mathcal
    F_i$-free, and $\tau(\mathbb I)$ is a \yes{} instance of $\MMSNP(\mathcal
    F)$.

    Conversely, if $\tau(\mathbb I)$ is a \yes{} instance of $\MMSNP(\mathcal
    G)$, then $\tau(\mathbb I)$ has a $\mathcal G_j$-free expansion for some
    $j\in[n]$, witnessed by $s\colon I' \to [d_j]$. Since every copy of $\mathbb
    G$ in $\tau(\mathbb I)$ is $\mathcal G_j$-free, for each $\mathbb G
    \in\mathcal G_j^*$, the same mapping $s$ is a homomorphism from $\mathbb I'$
    to $\mathbb T_j$. By \cref{lem:kun_girth}, we then have $\mathbb I \to
    \mathbb T_j \to \mathbb T$.
\end{proof}

This concludes the proof of PMMSNP--PCSP correspondence.

\section{Forbidding monochromatic cliques}
In this section we provide a complete complexity classification for Promise
MMSNPs defined by forbidding monochromatic cliques. Recall that $\templ{c}{k}$
denotes the family of all $k$-cliques such that every two vertices in the same
clique have the same colour out of $\{1, \dots, c\}$.

\dichotomy*

\Cref{th:mainequivalence} asserts that we can study Promise CSPs defined by
suitable relations instead of Promise MMSNPs. Concretely, for $c, k, \ell \ge
1$, let $\urel{c}{k}{\ell} \subseteq [c]^\ell$ be the relation consisting of all
tuples $(a_1, \dots, a_\ell)$ such that $\forall{a \in [c]}:~\lvert\{i : a_i =
a\}\rvert < k$. Observe that $\urel{d}{\ell}{\ell}$ is isomorphic to
$\NAE[d]{\ell}$.

\Cref{th:mainequivalence} applied to $(\templ{c}{k}, \templ{d}{\ell})$ provides
the structures $\mathbb S$ and $\mathbb T$. Following the construction in
\cref{sec:equivalence} (in fact, even the reasoning similar to \cref{ex:AIP}
suffices), the reader may verify that these structures correspond to
$\urel{c}{k}{\ell}$ and $\urel{d}{\ell}{\ell}$ respectively --- the extra $\sim$
relation is full in both structures, and can be safely discarded without
affecting the complexity of $\PCSP(\mathbb S, \mathbb T)$.

\cliqueequivalence*

The rest of this section is organized as follows. First, we want to characterise
the containment $\templ{c}{k} \to \templ{d}{\ell}$ in terms of the relation
between the numbers $c, d, k$, and $\ell$. Next, the proof of the tractability
part of the dichotomy is presented. Finally, we prove the hardness part, which
is incidentally the most technical part: it requires a basic introduction to
Fourier analysis and related notions.

\subsection*{Understanding containment}

\begin{lemma}
    Let $c,d,k,\ell$ be positive integers. Then $\templ{c}{k} \to
    \templ{d}{\ell}$ if and only if $k \leq \ceil{\ell/\ceil{c/d}}$.
\end{lemma}
\noindent We call this inequality the \emph{containment condition} for $c,d,k,\ell$.

\begin{proof}
    Suppose that $k \leq \ceil{\ell/\ceil{c/d}}$. Fix any mapping
    $r\colon[c]\to[d]$ such that $|r^{-1}(i)|\in \{\floor{c/d}, \ceil{c/d}\}$
    for every $i\in[d]$. Every $r$-preimage of every monochromatic $\mathbb
    K_{\ell}$ contains at least $\ceil{\ell/\ceil{c/d}}$ vertices of the same
    colour, so it contains a monochromatic copy of $\mathbb K_{k}$, which
    implies that $r$ is a recolouring.
    
    For the opposite direction, let $r\colon [c]\to[d]$ be an arbitrary
    recolouring from $\templ{c}{k}$ to $\templ{d}{\ell}$. By the pigeonhole
    principle, there must exist $i \in [d]$ such that $|r^{-1}(i)| \ge
    \ceil{c/d}$. Every $r$-preimage of $(\mathbb K_{\ell}, v \mapsto
    i)\in\templ{d}{\ell}$ contains
    \begin{itemize}
        \item at least $\ceil{\ell/\ceil{c/d}}$ vertices of the same colour, and
        \item a monochromatic copy of $\mathbb K_k$, since $r$ is a recolouring.
    \end{itemize}
    Consequently $k \leq \ceil{\ell/\ceil{c/d}}$.
\end{proof}

In particular, observe that if $c,d,k,\ell > 0$ satisfy the containment
condition, then $k \le \ell$.

\subsection*{Tractability}

We shall prove the first part of \cref{th:dichotomy}.

\begin{theorem}
    Let $c,d,k,\ell$ be positive integers satisfying the containment condition
    such that $\ell \ge c(k-1)$. Then $\PMMSNP(\templ{c}{k}, \templ{d}{\ell})$
    is in \P.
\end{theorem}
\begin{proof}
    By \cref{lemma:clique-equivalence}, it suffices to prove that
    $\PCSP(\urel{c}{k}{\ell}, \urel{d}{\ell}{\ell})$ is in \P. If $\ell >
    c(k-1)$, then the left-hand side relation is empty, so the problem is
    trivially tractable. Similarly if $d = 1$, then the right-hand side relation
    is empty, so the problem is yet again trivial. Thus assume that $\ell =
    c(k-1)$ and $d \geq 2$. We will argue that
    \[
        \urel{c}{k}{\ell} \to \kinl{(k-1)}{\ell} \to \NAE[2]{\ell} \to \urel{d}{\ell}{\ell}.
    \]
    Then, tractability follows by \Cref{obs:sandwich} from the tractability of
    \[
    \PCSP\left(\kinl{(k-1)}{\ell}, \NAE[2]{\ell}\right),
    \]
    which, in turn, follows from the techniques of~\cite{BG21_pcsps}, namely
    AIP.\@ Let us therefore show the required homomorphisms.
    \begin{description}
        \item[$\urel{c}{k}{\ell} \to \kinl{(k-1)}{\ell}$] Since $\ell =
            c(k-1)$, it follows that each tuple of $\urel{c}{k}{\ell}$ consists
            of exactly $k - 1$ appearances of each of the $c$ colours in our
            domain. Mapping one fixed colour to 1 and the rest to 0 gives rise
            to a homomorphism to $\kinl{(k-1)}{\ell}$.
        \item[{${\NAE[2]{\ell}} \to \urel{d}{\ell}{\ell}$}] Mapping 0 and 1 to
            any two distinct colours in the domain of $\urel{d}{\ell}{\ell}$
            suffices. (Such colours exist as $d \geq 2$.)\qedhere
    \end{description}
\end{proof}

\subsection*{Hardness}

For the second part of \cref{th:dichotomy}, we prove the following.

\begin{theorem}\label{th:mainhardness}
    Let $c, d, k, \ell$ be positive integers satisfying the containment
    condition such that $\ell < c(k-1)$. Then $\PMMSNP(\templ{c}{k},
    \templ{d}{\ell})$ is \NP-hard, assuming the Rich 2-to-1 Conjecture.
\end{theorem}

Braverman et al.~\cite{multislices} recently presented a yet different proof of
conditional hardness for $\PCSP(\mathbb K_3, \mathbb K_\ell)$ using their
analysis of the \emph{invariance principle} for functions on multislice domains, and reducing from
the Rich 2-to-1 Conjecture.\footnote{Earlier, Guruswami and
Sandeep~\cite{GS2020agc} proved this under the weaker $d$-to-1 Conjecture, which is claimed to have been resolved by~\cite{dto1proof,2to1openai}.
However, we were not able to extend their techniques to our setting.} We follow
their approach to derive hardness for the PCSPs of the following form:

\blackbox*

The proof of this result is rather technical, so we will first show how it
implies \cref{th:mainhardness}, and then devote the remaining part of this paper
to proving \cref{th:blackbox}.

\begin{proof}[Proof of \cref{th:mainhardness}]   
    Fix $c, d, k, \ell$ such that $\ell < c(k-1)$. First, note that by the
    containment condition
    \[
        k \le \ceil*{\frac{\ell}{\ceil*{\frac{c}{d}}}} \le \ell,
    \]
    hence $\ell < c(k-1) \le c(\ell-1)$. This means that $c, d, k, \ell \ge 2$.

    By \cref{lemma:clique-equivalence}, it suffices to show that
    $\PCSP(\urel{c}{k}{\ell}, \urel{d}{\ell}{\ell})$ is \NP-hard, assuming the
    Rich 2-to-1 Conjecture. We are going to apply \cref{th:blackbox}.
    
    Let $R \coloneqq \urel{c}{k}{\ell}$ which is a symmetric relation of arity
    $\ell$. Since $\ceil{\ell/c} \le 1 + \ell/c < k$, it is non-empty.

    Recall that $R$ consists of all tuples of length $\ell$, using only elements
    in $[c]$, where every element appears less than $k$ times. If we start from
    some arbitrary tuple of $R$, and replace (one of) the most frequent colours
    with (one of) the least frequent colours, we can reach any tuple where each
    colour appears either $\floor{\ell/c}$ or $\ceil{\ell/c}$ times. This
    implies that every tuple in $R$ can be reconfigured to the uniquely chosen
    ``almost balanced'' tuple in $R$.

    Therefore $R \neq \emptyset$ is a symmetric reconfigurable relation. On the
    other hand, $\urel{d}{\ell}{\ell} = \NAE[d]{\ell}$. We conclude the proof by
    applying \cref{th:blackbox}.
\end{proof}

\subsection{What about edge-coloured cliques?}
Some results from \cref{sec:equivalence} can be transferred to edge-colouring problems forbidding monochromatic cliques.
In order not to confuse with vertex-coloured clique families, we will use the notation $\mathcal E_\ell^{(d)}$ for a family of copies of $K_\ell$ monochromatically edge-coloured with $d$ colours.
Correspondingly, the pair $(\mathcal E_k^{(c)}, \mathcal E_\ell^{(d)})$ defines the following promise problem:

\vspace{0.2cm}
\begin{center}
    \parbox{0.85\textwidth}{
    \centering
    \textit{Given a graph $G$, decide whether $G$ has an edge-colouring with $c$ colours that omits a monochromatic $K_k$ or any edge-colouring of $G$ with $d$ colours contains a monochromatic $K_\ell$.}
    }
\end{center}
\vspace{0.2cm}

The notion of a recolouring between two edge-coloured families of cliques is defined similarly to the vertex-colouring problems and it also characterizes the containment~\cite{gmsnp_containment}.
Given two families $\mathcal E_k^{(c)}, \mathcal E_\ell^{(d)}$, we write $\mathcal E_k^{(c)} \to \mathcal E_\ell^{(d)}$ for the recolouring between the two problems.
Let $R(k, c)$ stand for the $c$-colour diagonal Ramsey number, i.e., the smallest number $N$ such that any edge-colouring of $K_N$ with $c$ colours contains a monochromatic $K_k$.

\begin{lemma}\label{lemma:edge_containment}
    Let $c,d,k,\ell$ be positive integers. Then $\mathcal E_k^{(c)}\to \mathcal E_\ell^{(d)}$ if and only if $\ell \geq R(k, \lceil c/d \rceil)$.
\end{lemma}
\begin{proof}
    Fix any $r\colon [c] \to [d]$ such that $r^{-1}(i) \in \{\lfloor c/d \rfloor, \lceil c/d \rceil\}$ for every $i\in[d]$.
    Since $\ell \geq R(k, \lceil c/d \rceil)$, any edge-colouring of $K_\ell$ in $\lceil c/d \rceil$ has a monochromatic $K_k$.
    In other words, every $r$-preimage of every monochromatic $K_\ell$ contains a monochromatic $K_k$, so $r$ is a recolouring.

    Let $r\colon [c]\to [d]$ be any recolouring and take $i\in[d]$ such that $r^{-1}(i) \geq \lceil c/d\rceil$ which must exist by pigeonhole principle.
    Since $r$ is a recolouring, every $r$-preimage of the $i$-monochromatic $K_\ell$ contains a monochromatic $K_k$.
    This implies that $\ell \geq R(k, |r^{-1}(i)|)$ and, in particular, $\ell \geq R(k, \lceil c/d \rceil)$.
\end{proof}

Denote by $\edgerel{c}{k}{\ell}$ the relation of arity $\binom{\ell}{2}$ that consists of all edge-colourings of $K_\ell$ with $c$ colours which omits a monochromatic copy of $K_k$.
Note that $\edgerel{d}{\ell}{\ell}$ is isomorphic to $\NAE[d]{\binom{\ell}{2}}$.
Promise edge-colouring problems with forbidding cliques also reduce to finite-domain PCSPs similarly to \Cref{th:mainequivalence}.
However, the reduction in the opposite direction which in the proof of \Cref{th:mainequivalence} relied on Sparse Incomparability Lemma fails for edge-colourings.

\begin{lemma}\label{lemma:edge_reduction}
    For every $c,d,k,\ell$ such that $\mathcal E_k^{(c)}\to \mathcal E_\ell^{(d)}$, the promise edge-colouring problem over $(\mathcal E_k^{(c)}, \mathcal E_\ell^{(d)})$ reduces in poly-time to $\PCSP(\edgerel{c}{k}{\ell}, \edgerel{d}{\ell}{\ell})$.
\end{lemma}

Although we could not provide a reduction in the other direction, \Cref{lemma:edge_reduction} still gives tractability of some promise edge-colouring problems, as the following example shows.
\begin{example}\label{example:k3vsk5}
    Take the families $\mathcal E_3^{(2)}$ and $\mathcal E_5^{(2)}$.
    As $R(3,1)=4$, we have $\mathcal E_3^{(2)} \to \mathcal E_5^{(2)}$, by \Cref{lemma:edge_containment}.
    Then, \Cref{lemma:edge_reduction} ensures that the corresponding edge-colouring problem is tractable if $\PCSP(\edgerel{2}{3}{5}, \edgerel{2}{5}{5})$ is.
    It is easy to verify that the only way to partition the edges of $K_5$ in two triangle-free sets is when every partition class induces a graph isomorphic to $C_5$.
    This means that every element of $\edgerel{2}{3}{5}$ has exactly five entries equal to each colour, and in particular, there is a homomorphism $\edgerel{2}{3}{5} \to \kinl{5}{10}$.
    Since $\edgerel{2}{5}{5}$ is just $\NAE[2]{10}$, we conclude that the tractability of $\PCSP(\edgerel{2}{3}{5}, \edgerel{2}{5}{5})$ follows from the tractability of $\PCSP(\kinl{5}{10}, \NAE[2]{10})$ by \Cref{obs:sandwich}.
\end{example}

\section{PCSP hardness}\label{sec:hardness}
This section is devoted to the proof of \cref{th:blackbox}. As mentioned
earlier, it is a special case of a more general hardness result that we prove:
the restrictions on $R$ shall be relaxed in two ways.

\subsection*{Symmetric $\to$ \balanced}
It suffices that $R$ possesses a sufficient level of \emph{balance}, a concept introduced by Nakajima and \v{Z}ivn\'{y} \cite{functional}; we call it \emph{\balanced} for clarity.
\begin{definition}[Based on \cite{functional}]
    A relation $R \subseteq A^r$ is \emph{\balanced} if there exists a matrix whose columns are the tuples of $R$, where each tuple of $R$ appears at least once, and whose rows are permutations of each other.
\end{definition}

For instance, every symmetric relation is \balanced\ (in fact, it suffices merely to be preserved by the action of a \emph{transitive} permutation group). As another example, the edge relation of any strongly connected digraph --- in fact, of any graph whose every connected component is strongly connected --- is \balanced.

\subsection*{Reconfigurable $\to$ BKLM-connected}
The \emph{reconfigurability} shall be relaxed to \emph{connectedness} of Braverman, Khot, Lifshitz, and Minzer~\cite{multislices}; we call it \emph{BKLM-connectedness} for clarity.
\begin{definition}[Based on~\cite{multislices}]
    Given a relation $R \subseteq A^r$ and $1 \le i \le r$, we define $\mathbb
    G^R_i$ to be the bipartite graph whose left side consists of all $a \in A$ that appear in the $i$-th position of some tuple in $R$, whose right side consists of all tuples in $A^{r-1}$ that can be extended to a tuple in $R$ by inserting some $a \in A$ into the $i$-th position, and where a left vertex and a right vertex are connected if they form a tuple in $R$ together.
    
    In other words, if we define $\bar{a_{\neg i}} = (a_1, \ldots, a_{i - 1}, a_{i + 1}, \ldots, a_n)$ for $\bar{a} = (a_1, \ldots, a_n)$, then
    \[
    \mathbb G_i^R \coloneqq (X_i^R \cup Y_i^R; E_i^R)
    \]
    where
    \begin{align*}
        X_i^R \coloneqq \{ a_i \mid \bar{a} \in R \} \quad\text{ and }\quad Y_i^R \coloneqq \{ \bar{a_{\neg i}} \mid \bar{a} \in R \} \quad\text{ and }\quad
        E_i^R \coloneqq \{ (a_i, \bar{a_{\neg i}}) \mid \bar{a} \in R \}.
    \end{align*}

    We say that $R$ is \emph{BKLM-connected} if $\mathbb    G^R_i$ is connected for each $i \in [r]$.
\end{definition}
\begin{proposition}
    Let $R \subseteq A^r$ be a reconfigurable relation. Then $R$ is BKLM-connected.
\end{proposition}
\begin{proof}
    Fix $1 \le i \le r$, and consider the bipartite graph $\mathbb G_i^R$. Since
    this graph has no isolated vertices by construction, it is connected if and
    only if its \emph{line graph}\footnote{The line graph of $\mathbb{G}_i^R$
    has as vertices the edges of $\mathbb{G}_i^R$ (i.e.~tuples of $R$), and two
    edges are connected whenever one of their endpoints coincides (i.e.~two
    tuples are connected if they agree on the $i$-th position or on all other positions). } is
    connected.
    
    Consider now the \emph{reconfiguration graph} of $R$: the vertices of this
    graph are the tuples of $R$; two tuples are connected whenever they differ
    in precisely one coordinate. Clearly, the reconfiguration graph is a
    subgraph of the line graph of $\mathbb G_i^R$, and they have the same vertex
    sets. Since the former is connected by the assumption, the latter must be
    connected as well, which concludes the argument.
\end{proof}

We are ready to state our most general hardness result.
\begin{restatable}{theorem}{mainmainhardness}\label{th:mainmainhardness}
    Let $\emptyset \neq R \subseteq A^r$ be a \balanced\ BKLM-connected relation.
    For any $d$ such that $R \to \NAE[d]{r}$, the problem $\PCSP(R, \NAE[d]{r})$ is \NP-hard, assuming the Rich 2-to-1 Conjecture.
\end{restatable}

\subsection*{Dictatorship testing}

As discussed in the \nameref{sec:intro}, our approach extends the pipeline of \cite{multislices} that converts any perfect-completeness dictatorship test into a hardness of approximately solving satisfiable instances of BKLM-connected CSPs, assuming the Rich 2-to-1 Conjecture, to Promise CSPs.

For simplicity, we give a specialised definition of a dictatorship test tailored to PCSP templates with a single relation symbol in the language (which are in the scope of \cref{th:mainmainhardness}).
The most general definition is deferred to Appendix.
For the following, we need to introduce \emph{regularity} for functions --- a proxy for being ``far from any dictator''.
Intuitively, we will say that a function $f \colon A^n \to B$ is $(k, \tau)$-regular if no coordinate has a large influence on its output; the formal definition is deferred to Appendix. 

For the following definition, if $\Omega$ is a probability distribution over tuples of arity $r$, then we let $\Omega^n$ be the distribution over matrices with $n$ columns and $r$ rows, where each column is drawn independently according to $\Omega$. Furthermore $\transpose{x_1}{x_r}$ denotes the matrix whose \emph{rows} are $x_1, \ldots, x_r$.

\begin{definition}[Based on {\cite[Definition 6.6]{multislices}}]
    Let $(\mathbb A, \mathbb B)$ be a PCSP template over the relational language $\{R\}$.
    A distribution $\Omega$ over $R^\mathbb A$ is called a \emph{dictatorship test} for $(\mathbb A, \mathbb B)$ if it satisfies the following property.\footnote{To be precise, this is a special case of what in the literature would be called a $(1,s)$-dictatorship test with perfect completeness and soundness $s$.}
    \begin{description}
        \item[(soundness)] $\exists{s < 1}~ \exists{k, \tau > 0}~ \forall{(k,\tau)\text{-regular } f \colon A^n \to B}:$
        \[
        \Pr_{\transpose{\bar x^{(1)}}{\bar x^{(r)}} \sim \Omega^n} \left[\left(f(\bar x^{(1)}), \dots, f(\bar x^{(r)})\right) \in R^{\mathbb B} \right] \le s.
        \]
    \end{description}
\end{definition}

The existence of a dictatorship test with a BKLM-connected support facilitates a reduction from the Rich 2-to-1 Conjecture to the PCSP.
The proof of this fact is heavily analytical; in
particular, it relies on the \emph{invariance principle} on \emph{multislice}
domains developed in \cite{multislices}.

\begin{theorem}[Special case of~\cref{th:dictatorship-to-hardness}; Based on {\cite[Theorem 6.7]{multislices}}]\label{th:single-dictatorship-to-hardness}
    Let $(\mathbb A, \mathbb B)$ be a PCSP template.
    If it has a dictatorship test $\Omega$ with a BKLM-connected support, then $\PCSP(\mathbb A, \mathbb B)$ is \NP-hard, assuming the Rich 2-to-1 Conjecture.
\end{theorem}

Thus deriving \cref{th:mainmainhardness} amounts to finding dictatorship tests for the templates $(R, \NAE{r})$ of interest.
Our result then takes the following form:
\begin{lemma}\label{lem:relaxed-hardness}
    Let $\emptyset \neq R \subseteq A^r$ be a \balanced\ BKLM-connected relation.
    For any $d$ such that $R \to \NAE[d]{r}$, there exists a dictatorship test for $(R, \NAE[d]{r})$ with a BKLM-connected support.
\end{lemma}

Our source of such a dictatorship test comes from the powerful noise inequalities of Mossel~\cite[Proposition 6.4]{mossel};
we state a variant of his result, tailored to our purposes:

\begin{restatable}{theorem}{mossel}\label{th:mossel}
    Let $\Omega$ be a distribution over $A^r$ whose support is a BKLM-connected relation.

    $\forall{\delta > 0}~\exists{k, \tau, \varepsilon > 0}~\forall{(k,\tau)\text{-regular } g \colon A^n \to \{0,1\}}$
    \[
    \forall{i \in [r]}:\E_{\overline x \sim \Omega_i^n}[g(\overline x)] \ge \delta \implies \E_{\transpose{\bar x^{(1)}}{\bar x^{(r)}} \sim \Omega^n}\left[ \prod_{i=1}^r g(\overline x^{(i)})\right ] \ge \varepsilon.
    \]
\end{restatable}
\noindent For completeness, its proof is provided in the Appendix.
\begin{proof}[Proof of \cref{lem:relaxed-hardness}]
    Since $R$ is \balanced, there exists a distribution $\Omega$ fully supported in $R$ such that all marginals of $\Omega$ are equal.

    Let $k, \tau, \varepsilon$ be obtained from \cref{th:mossel} for $\delta = 1/d$.
    We verify that $\Omega$ defines a valid dictatorship test for $s \coloneqq 1-\varepsilon$.
    
    Fix a $(k, \tau)$-regular $f \colon A^n \to B$ where $B \coloneqq [d]$.
    For each $b \in B$, let $f^b \colon A^n \to \{0,1\}$ denote the $b$-indicator of $f$, that is, $f^b \coloneqq \textbf 1[f = b]$.
    Note that each $f^b$ is also $(k, \tau)$-regular (we prove this fact formally in Appendix).
    Furthermore, there exists $b \in B$ such that
    \[
    \E[f^b] \ge 1/d.
    \]
    We apply \cref{th:mossel} to this $f^b$, and conclude that
    \[
    \E_{\Omega^n}\left[ \prod_{i=1}^r f^b(\bar x^{(i)}) \right] \ge \varepsilon.
    \]
    Observe that $(a_1, \dots, a_r) \in \NAE[d]{r}$ if and only if for every $b \in B$ some $a_i \neq b$. In particular 
    \[
    \left(f(\bar x^{(1)}), \dots, f(\bar x^{(r)})\right) \in \NAE[d]{r} \iff \forall{b \in B}:~f^b(\bar x^{(1)})\cdot \ldots \cdot f^b(\bar x^{(r)}) = 0.
    \]
    Consequently
    \[
    \Pr_{\Omega^n} \left[\left(f(\bar x^{(1)}), \dots, f(\bar x^{(r)})\right) \in \NAE[d]{r}\right] \le 1 - \E_{\Omega^n}\left[ \prod_{i=1}^r f^b(\bar x^{(i)}) \right] \le 1 - \varepsilon = s
    \]
    as required.
\end{proof}

\section{Conclusions}\label{sec:conclusions}
In the present paper, we initiated the study of Promise MMSNP problems. For the
most fundamental case, forbidding monochromatic cliques, we obtained a
complexity dichotomy under the Rich 2-to-1 Conjecture.
This class of problems
includes the notorious Approximate Graph Colouring,
which admits a poly-time reduction~\cite{GS2020agc} from the weaker $d$-to-1 Conjecture; \cite{dto1proof, 2to1openai} are claiming to have resolved this conjecture.
Some other special cases of our Promise MMSNPs follow from even simpler existing unconditional hardness results: for example, if the forbidden cliques only use 2 colours, the hardness can be
derived from the dichotomy for symmetric Boolean PCSPs
by~\cite{ficak2019symmetric}.

Hence the most natural question is whether ``richness''
is necessary for the Promise MMSNPs studied in this paper. In other words,
\emph{do Promise MMSNPs over forbidden monochromatic cliques admit the dichotomy
under the 2-to-1 Conjecture?}

These Promise MMSNPs are tightly connected with certain \emph{balanced}
hypergraph colouring problems. More concretely, consider $\PCSP(R, \NAE{r})$
where $R$ constitutes the promise colouring. If the promise colouring is
\emph{perfectly balanced}, that is, each colour appears the same number of times
in every hyperedge, then a proper ($\NAE{r}$) colouring can be computed
efficiently (see \cite{GL2018balancedrainbow}). For multiple special cases, this
was proved to be the only tractable case, e.g.,~\cite{GL2018balancedrainbow,
GS2018rainbow}. Our result can also be viewed as another special case where a
perfect balance of the colours characterises tractability.

There are also some interesting connections with the \emph{topological approach
to PCSPs}. Some of the previous results on rainbow-free colouring of hypergraphs
are based on topological combinatorics~\cite{ABP20_rainbow}; while our approach
is Fourier-analytic, the notion of \emph{reconfigurability} is (perhaps
unsurprisingly, given the name) in some sense topological. Indeed, many of the
recent applications of topology to
(P)CSPs~\cite{KOWZ23_topology,FNOTW25_topology,MO25_topology,AFOTW25_topology}
have used \emph{hom-complexes}, a vast generalisation of the reconfiguration
graph.  (It is worth noting that our \Cref{th:blackbox} implies all the hardness
results that we know of that follow via the topological approach to PCSP,
although under the additional hypothesis of the Rich 2-to-1 Conjecture.)
Such connections between topological and Fourier-analytic approaches
to PCSP have been noticed before, e.g.~a connection between the work
of~\cite{KOWZ23_topology} and Fourier analysis via winding numbers due to Mottet
(personal communication). Finding a principled reason why these connections
appear may be worth investigating.

Besides the forbidden vertex-coloured patterns, there has been some research on
the edge-coloured ones. It would be interesting to obtain similar results for
promise problems of forbidden monochromatic cliques of the latter type as well.
The forbidden edge-coloured patterns give rise to a larger than MMSNP fragment
of Existential Second-Order Logic called \emph{GMSNP}, see~\cite{obda}. Unlike
MMSNP, the computational dichotomy for GMSNP remains unresolved; however, it
lies within the scope of the Bodirsky-Pinsker conjecture as well as
MMSNP~\cite{asnp}. Recently, the containment problem for two GMSNP formulas was
shown to be decidable, and also dependent on the existence of a
recolouring~\cite{gmsnp_containment}. Deciding whether a given graph has an edge
2-colouring that avoids monochromatic triangles is a classic \NP-complete
problem~\cite{np_completeness}. It remains \NP-hard for an arbitrary (fixed)
clique size~\cite{burr}, and for an arbitrary (fixed) number of
colours~\cite{edge_colourings_np_complete}. Nevertheless, its promise extensions
present new obstacles, compared to the vertex-coloured cliques. In particular,
\begin{itemize}
    \item Characterising the existence of a recolouring must depend on the
        Ramsey numbers, which are very hard to compute.
    \item  It remains unclear how to obtain a reduction from the finite PCSP to the promise edge-colouring problem. All known reductions from CSPs to edge-colouring problems depend on the template of the CSP, i.e., different templates produce different reductions.
\end{itemize}

\section*{Acknowledgements}
We would like to thank Andrei Krokhin and Amey Bhangale for useful discussions
and the anonymous reviewers for their valuable suggestions. Furthermore, we sincerely thank the
organisers of the 2025 CSP World Congress, where this project was initiated.

\section*{AI use statement}
The research was done without the help of AI. The assistance of ChatGPT-5.5 was used for spell-checking and correcting typos. All text in this manuscript is human-written, and the authors are fully responsible for the entire contents of the manuscript.


\bibliographystyle{amsplain}
\bibliography{bibliography}

\appendix

\section{New minor conditions}\label{app:newMinorConditions}
Following the algebraic approach to PCSPs \cite{Bible}, we state a hardness condition in terms of \emph{polymorphisms} that follows from \cref{th:mainmainhardness}.

\begin{definition}
    Let $\sigma$ be a relational signature and $(\mathbb A, \mathbb B)$ be a PCSP template. A function $f \colon A^n \to B$ is called an $n$-ary \emph{polymorphism} of $(\mathbb A, \mathbb B)$ if it is a homomorphism from $\mathbb A^n \to \mathbb B$, that is, a function that satisfies, for each $T \in \sigma$ and all tuples $\bar a^{(1)}, \dots, \bar a^{(n)} \in T^\mathbb A$,
    \[
    \left( f(a^{(1)}_1, \dots, a^{(n)}_1), \dots, f(a^{(1)}_{\ar(R)}, \dots, a^{(n)}_{\ar(R)}) \right) \in T^\mathbb B.
    \]
\end{definition}

\cite{Bible} devised an elegant hardness condition: if a PCSP template does not admit an \emph{Ol\v{s}\'ak} polymorphism, then it is \NP-hard.
This followed from the hardness of approximate 3-uniform hypergraph colouring \cite{DRS05}, i.e., $\PCSP(\NAE[2]{3}, \NAE[d]{3})$.
Their approach can be extended to PCSPs of our form, that is, $\PCSP(R, \NAE[d]{r})$, to yield an infinite family of \emph{special} polymorphisms whose lack makes PCSP hard (assuming the Rich 2-to-1 Conjecture).

To formalise this idea, fix a relation $R \subseteq [c]^r$ and enumerate it
\[
R = \left\{ \bar a^{(1)}, \dots, \bar a^{(m)} \right\}.
\]
We will say that a function $f \colon A^m \to B$ is \emph{$R$-invariant} if it satisfies
\begin{align*}
    &\ f(x_{a^{(1)}_1}, \dots, x_{a^{(m)}_1}) = \\
    = &\ f(x_{a^{(1)}_2}, \dots, x_{a^{(m)}_2}) = \\
    = &\ \dots \\
    = &\ f(x_{a^{(1)}_r}, \dots, x_{a^{(m)}_r}) 
\end{align*}
for all $\bar x \in A^c$.

\begin{example}
    Let $C_{2,3} \subseteq \{1,2\}^5$ be the relation obtained by closing the tuple $(1,1,2,2,2)$ under cyclic shifts and swapping the roles of $1$ and $2$ (we have already mentioned this relation in the \cref{sec:intro}).
    Note that $C_{2,3}$ is \balanced\ and BKLM-connected.

    A function $f \colon A^{10} \to B$ is $C_{2,3}$-invariant if it satisfies\footnote{We rename $x_1, x_2$ to $x, y$, and omit some commas for readability.}
    \begin{align*}
    &\ f(xxxyy,yyyxx) = \\
    = &\ f(xxyyx,yyxxy) = \\
    = &\ f(xyyxx,yxxyy) = \\
    = &\ f(yyxxx,xxyyy) = \\
    = &\ f(yxxxy,xyyyx)
    \end{align*}
    for all $x, y \in A$.
\end{example}

\begin{corollary}
    Let $(\mathbb A, \mathbb B)$ be a PCSP template.
    If there exists an \balanced\ BKLM-connected $R$ such that $(\mathbb A, \mathbb B)$ has no $R$-invariant polymorphism, then $\PCSP(\mathbb A, \mathbb B)$ is \NP-hard, assuming the Rich 2-to-1 Conjecture.
\end{corollary}
\begin{proof}
    Our proof follows \cite[Theorem 6.2]{Bible}.
    Fix an \balanced\ BKLM-connected $R \subseteq [c]^r$ that satisfies the assumptions; enumerate
    \[
    R = \left\{ \bar a^{(1)}, \dots, \bar a^{(m)} \right\}.
    \]
    
    Consider the hypergraph $\mathbb F$ defined as follows\footnote{The hypergraph $\mathbb F$ is the free structure of the polymorphisms of $(\mathbb A, \mathbb B)$ generated by $R$; see~\cite{Bible} for the definition of a free structure.}: the vertices of $\mathbb F$ are all $c$-ary polymorphisms of $(\mathbb A, \mathbb B)$, and $r$ such vertices $g_1, \dots, g_r$ are connected with a hyperedge if there exists an $m$-ary $f$ such that
    \begin{align*}
        g_1(x_1, \dots, x_c) &= f(x_{a^{(1)}_1}, \dots, x_{a^{(m)}_1}) \\
        g_2(x_1, \dots, x_c) &= f(x_{a^{(1)}_2}, \dots, x_{a^{(m)}_2}) \\
        &\dots\\
        g_r(x_1, \dots, x_c) &= f(x_{a^{(1)}_r}, \dots, x_{a^{(m)}_r})
    \end{align*}
    for all $x_1, \dots, x_c \in A$.
    
    Observe that $\mathbb F$ does not contain a hyperedge $(g_1, \dots, g_r)$ with $g_1 = \dots = g_r$ from the assumption.
    Hence $\mathbb F$ is colourable with $K$ colours, where $K = |F|$ is the number of $c$-ary polymorphisms of $(\mathbb A, \mathbb B)$, that is, $\mathbb F \to ([K]; \NAE[K]{r})$.
    \cite[Lemma 4.4]{Bible} then gives a minion homomorphism from $(\mathbb A, \mathbb B)$ to $(R, \NAE[K]{r})$.
    Finally, by \cite[Theorem 3.1]{Bible} we have that $\PCSP(R, \NAE[K]{r}) \le_p \PCSP(\mathbb A, \mathbb B)$, and the former is \NP-hard assuming the Rich 2-to-1 Conjecture (\cref{th:mainmainhardness}).
\end{proof}

\begin{example}
    Recall \cref{rem:nice-condition} that characterises relations possessing an \balanced\ and reconfigurable subrelation.
    We can use this characterisation to derive a polymorphic hardness condition.
    
    Suppose that there exists a (finite) matrix filled with symbols $x_1, x_2, \dots$ such that
    \begin{itemize}
        \item the rows are permutations of each other,
        \item for every column $c$ (other than the first one) there exists a column $c'$ to its left which differs from it in at most one position, and
        \item no polymorphism of $(\mathbb A, \mathbb B)$ applied row-wise outputs a constant tuple independent of the values of $x_i$.
    \end{itemize}
    Then $\PCSP(\mathbb A, \mathbb B)$ is \NP-hard, assuming the Rich 2-to-1 Conjecture.
\end{example}

\section{Fourier analysis: basics}
The main goal of this section is to define regular functions formally, and mention some of their properties used later in the proofs.
We only introduce the notions from Fourier analysis relevant to our results; see \cite{odonnell} for a more extensive exposition of this area.

Given a probability distribution $\nu$ over a finite set $A$ and a function $f \colon A^n \to \mathbb R^m$,
we consider its \emph{Efron-Stein decomposition}\footnote{We refer the reader to \cite[Section 8.3]{odonnell} for more details.} with respect to $\nu^n$:
\[
f = \sum_{S \subseteq [n]} f^{=S}
\]
where each $f^{=S} \colon A^n \to \mathbb R^m$ gives the ``contribution of $S$ to $f$''.

The influence of coordinate $j \in [n]$ is defined as\footnote{We sometimes omit the second parameter $\nu$ if it is not relevant or clear from the context}
\[
\Inf[\nu]{f, j} \coloneqq \E_{\bar x \sim \nu^n}\left[ \left\lVert f(\bar x) - E_j f(\bar x)\right\rVert_2^2 \right]
\]
where $E_j$ is the averaging operator over coordinate $j$ defined as $E_j f (\bar x) \coloneqq \E_{x_j \sim \nu}[f(\bar x)]$.
The influence of $j$ can also be thought of as the total ``contribution of all sets $S$ containing $j$ to $f$'':
\begin{observation}[Folklore]
    $\Inf[\nu]{f, j} = \sum_{S \ni j} \lVert f^{=S} \rVert_2^2$.
\end{observation}

The \emph{low-degree part} of $f$ gives its regularity properties: given $k \ge 0$, we denote
\[
f^{\le k} \coloneqq \sum_{|S| \le k} f^{=S}.
\]
Then the degree-$k$ influence of coordinate $j \in [n]$ is
\[
\lowInf[\nu]{f,j,k} \coloneqq \Inf[\nu]{f^{\le k}, j}.
\]

\begin{observation}[Folklore]
    $\lowInf[\nu]{f,j,k} = \sum_{S \ni j, |S| \le k} \lVert f^{=S} \rVert_2^2$.
\end{observation}

In dictatorship tests, we study functions with a discrete finite range; we fit these into the above framework as follows.
Given $f \colon A^n \to B$ where $B$ is a finite set, we identify the elements of $B$ with the vertices of the $|B|$-dimensional simplex, and thus treat the range of $f$ as $\mathbb R^{|B|}$.

\begin{definition}[Based on {\cite[Definition 6.5]{multislices}}]
    Suppose that $A, B$ are finite sets, and $\Omega$ is a distribution over $A^r$.
    We say that $f \colon A^n \to B$ is \emph{$(k, \tau)$-regular w.r.t. $\Omega$} if for all $i \in [r]$ and $j \in [n]$ we have that $\lowInf[\Omega_i]{f,j,k} \le \tau$.
\end{definition}

It is immediate that ``gluing'' range values cannot break regularity:
\begin{observation}
    If $f \colon A^n \to B$ is $(k, \tau)$-regular w.r.t. $\Omega$, then so is $f^b$, for any $b \in B$.
\end{observation}

Low-degree influences can be related to normal influences of a ``smoothed'' version of the original function; this connection is standard, and is formalised via Bonami-Beckner noise operator:
Given $\gamma \in [0,1]$, the noise operator $T$ parameterised by $1-\gamma$ maps any function $f \colon A^n \to \mathbb R$ to
\[
A^n \ni \bar x \mapsto \E_{\bar y} [f(\bar y)]
\]
where each $y_i = x_i$ with probability $1-\gamma$, and with probability $\gamma$ it is sampled uniformly at random from $A$.
\begin{lemma}[Folklore]\label{lem:low-degree-and-noise}
    Let $\gamma \in (0,1)$; denote by $T$ the Bonami-Beckner noise operator parameterised by $1-\gamma$.
    Then for any $f \colon A^n \to [0,1]$
    \begin{itemize}
        \item $\lowInf{Tf, j, k} \le \lowInf{f, j, k}$, and
        \item $\Inf{Tf, j} \le \lowInf{f, j, k} + (1-\gamma)^{2k}$.
    \end{itemize}
\end{lemma}

\section{Dictatorship tests for PCSPs}
We present the full generalisation of the ``dictatorship test to hardness'' pipeline of \cite{multislices} to Promise CSPs.

\begin{definition}[Based on {\cite[Definition 6.5]{multislices}}]
    For a collection of distributions $\{\Omega(t)\}_{t \in \mathcal T}$ and a measure $\chi \colon \mathcal T \to [0,1]$,
    we say that $f \colon A^n \to B$ is \emph{$(\gamma, k, \tau)$-regular w.r.t. $\{\Omega(t)\}_t$} if $f$ is $(k, \tau)$-regular w.r.t. $\Omega(t)$ with probability at least $1-\gamma$ over a random $t \sim \chi$.
\end{definition}

\begin{definition}[Based on {\cite[Definition 6.6]{multislices}}]
    Let $1 \ge c > s > 0$ and $(\mathbb A, \mathbb B)$ be a PCSP template over the relational language $\sigma$.
    A \emph{$(c,s)$-dictatorship test} for $(\mathbb A, \mathbb B)$ consists of
    \begin{enumerate}
        \item a probability distribution $\varrho$ over $\sigma$,
        \item a measurable set $\mathcal T$ and a probability distribution $\chi$ over $\mathcal T$, and
        \item a collection of distributions $\{\Omega(t, R)\}_{t \in \mathcal T, R \in \sigma}$ where each $\Omega(t, R)$ is over $A^{\mathrm{ar}(R)}$
    \end{enumerate}
    that satisfy the following properties
    \begin{description}
        \item[(completeness)] $\forall{f \colon A^n \to B}:$ if $f(\bar x) = h(x_i)$ for some $i \in [n]$ and some homomorphism $h$ from $\mathbb A$ to $\mathbb B$, then
        \[
        \Pr_{\substack{t \sim \xi, R \sim \varrho \\ \transpose{\bar x^{(1)}}{\bar x^{(r)}} \sim \Omega(t, R)^n}} \left[\left(f(\bar x^{(1)}), \dots, f(\bar x^{(r)})\right) \in R^{\mathbb B} \right] \ge c.
        \]
        \item[(soundness)] $\forall{\gamma, \varepsilon > 0}~ \exists{\tau,k, N > 0}~ \forall{f \colon A^n \to B}:$
        \begin{multline*}
        \quad\quad n \ge N ~\land~ \text{$f$ is $(\gamma, k, \tau)$-regular w.r.t. $\{\Omega(t, R)\}_{t,R}$} \implies \\
        \Pr_{\substack{t \sim \xi, R \sim \varrho \\ \transpose{\bar x^{(1)}}{\bar x^{(r)}} \sim \Omega(t, R)^n}} \left[\left(f(\bar x^{(1)}), \dots, f(\bar x^{(r)})\right) \in R^{\mathbb B} \right] \le s + \varepsilon.
        \end{multline*}
    \end{description}
\end{definition}

In his seminal work on approximability of CSPs, Raghavendra~\cite{raghavendra} discovered a ``dictatorship test to hardness'' pipeline for \emph{unsatisfiable} CSP instances.
\begin{theorem}[See \cite{raghavendra}]
    Let $\mathbb A$ be a relational structure, $1 > c > s$ and $\varepsilon > 0$.
    If $(\mathbb A, \mathbb A)$ has a $(c, s)$-dictatorship test, then the following problem is \NP-hard, assuming the Unique Games Conjecture.
    \vspace{.2cm}
    \begin{center}
    \begin{minipage}{0.8\textwidth}
        \centering
        Given an instance of $\CSP(\mathbb A)$, promised to admit an assignment satisfying $c$-fraction of all constraints, find an assignment satisfying $(s+\varepsilon)$-fraction of all constraints.
    \end{minipage}
    \end{center}
    \vspace{.2cm}
\end{theorem}

Later, \cite{multislices} extended this idea to \emph{satisfiable} CSP instances with an important restriction: the distributions in a dictatorship test must be BKLM-connected.

\begin{theorem}[See {\cite[Theorem 6.7]{multislices}}]
    Let $\mathbb A$ be a relational structure, $1 > s$ and $\varepsilon > 0$.
    If $(\mathbb A, \mathbb A)$ has a $(1, s)$-dictatorship test $\{\Omega(t, R)\}_{t,R}$ in which every $\Omega(t, R)$ has a BKLM-connected support, then the following problem is \NP-hard, assuming the Rich 2-to-1 Conjecture.
    \vspace{.2cm}
    \begin{center}
    \begin{minipage}{0.8\textwidth}
        \centering
        Given an instance of $\CSP(\mathbb A)$, promised to admit a solution, find an assignment satisfying $(s+\varepsilon)$-fraction of all constraints.
    \end{minipage}
    \end{center}
    \vspace{.2cm}
\end{theorem}

We note that their proof is oblivious to the fact that $\mathbb A = \mathbb B$ in the dictatorship test, and thus can be directly extended to Promise CSPs with minimal changes.
In fact, the authors of \cite{multislices} did exactly that for the Approximate Graph Colouring problem $\PCSP(K_3, K_d)$.
In view of this, we present the statement of the general PCSP extension here without the proof, for the sake of brevity.

\begin{theorem}[Based on {\cite[Theorem 6.7]{multislices}}]\label{th:dictatorship-to-hardness}
    Let $(\mathbb A, \mathbb B)$ be relational structures over the same language.
    If $(\mathbb A, \mathbb B)$ has a $(1, s)$-dictatorship test $\{\Omega(t, R)\}_{t,R}$ for some $s < 1$ in which every $\Omega(t, R)$ has a BKLM-connected support, then $\PCSP(\mathbb A, \mathbb B)$ is \NP-hard, assuming the Rich 2-to-1 Conjecture.
\end{theorem}

\section{Noise inequalities}

The dictatorship tests that we produce for PCSPs in~\cref{lem:relaxed-hardness} are simply uniform distributions over the left-hand side relations.
The way we prove that these distributions in fact define valid dictatorship tests is via correlated noise inequalities, to which this section is devoted.

For a distribution $\Omega$ over $A \times B$, we denote the correlation between the first and second coordinates of $\Omega$ with $\rho(A, B; \Omega) \in [0, 1]$.
Formally
\[
\rho(A, B; \Omega) \coloneqq \sup \left\{ \textbf{Cov}_\Omega[f, g] \mathrel{\big\vert} f \colon A \to \mathbb R, g \colon B \to \mathbb R, \textbf{Var}_\Omega[f] = \textbf{Var}_\Omega[g] = 1 \right\}.
\]

The reason BKLM-connected relations are interesting is due to the following
correlation bound that they provide:
\begin{lemma}[See {\cite[Lemma 2.9]{mossel}}]\label{lem:connected-correlation}
    Let $\Omega$ be a distribution over $A \times B$ whose support forms a
    connected bipartite graph. Then $\rho(A, B; \Omega) < 1$.
\end{lemma}
Abusing the notation, for a distribution $\Omega$ over $A_1 \times \dots \times A_r$ and two disjoint sets $I, J \subseteq [r]$, we write
\[
\rho\left(\prod_{i \in I} A_i, \prod_{j \in J} A_j; \Omega\right)
\]
for the correlation in the marginal of $\Omega$ to $I \cup J$ treated as a bivariate distribution over $(\prod_{i \in I} A_i) \times (\prod_{j \in J} A_j)$.
It is easy to see that removing elements from $I$ or $J$ cannot increase correlation:

\begin{observation}[Folklore]\label{obs:subset-correlation}
    Let $\Omega$ be a distribution over $A_1 \times \dots \times A_r$.
    Let $I' \subseteq I \subseteq [r]$ and $J' \subseteq J \subseteq [r]$ such that $I \cap J = \emptyset$. Then
    \[
    \rho\left(\prod_{i \in I'} A_i, \prod_{j \in J'} A_j; \Omega\right) \le \rho\left(\prod_{i \in I} A_i, \prod_{j \in J} A_j; \Omega\right).
    \]
\end{observation}

\begin{definition}
For a distribution $\Omega$ over $A_1 \times \cdots \times A_r$, we define its \emph{correlation
vector} $\tuple \rho \in [0,1]^{r-1}$ as
\[
\rho_i \coloneqq \rho(A_i, A_{i+1} \times \cdots \times A_r; \Omega).
\]
It is important to note the discrepancy with \cite{mossel} here: Mossel defines the correlation vector by computing the correlation between $A_1 \times \cdots \times A_i$ and $A_{i+1} \times \cdots \times A_r$.
\cref{obs:subset-correlation} implies that his correlations upper-bound our correlations.
\end{definition}

Finally, we need Gaussian stability measures. Let $\mathcal N$ be the
distribution function of a standard Gaussian. Mossel \cite{mossel} defines, for $\rho, \mu, \nu \in
[0,1]$,
\[
\Gamma_{\rho}(\mu, \nu) = \Pr_{X,Y}\left[ X \le \mathcal N^{-1}(\mu), Y \ge \mathcal N^{-1}(1-\nu) \right]
\]
where $(X, Y)$ are $\rho$-correlated\footnote{$(X, Y)$ are $\rho$-correlated if $Y = \rho X + \sqrt{1-\rho^2} Z$ where $Z$ is a standard Gaussian independent of $X$ and $Y$.} Gaussian variables. Note that $\Gamma_\rho$ is an increasing
function in $\mu$ and $\nu$. The only additional fact we need about $\Gamma_\rho$ is the
following.
\begin{observation}[See {\cite[Section 2.3]{dinur2006approxcol}}]\label{obs:positive-gamma}
    If $\mu > 0$ and $\rho < 1$, then $\Gamma_\rho(\mu, \mu) > 0$.
\end{observation}
The definition of $\Gamma$ is extended to more than 2 variables in a recursive
way. Namely, given $\tuple \rho \in [0,1]^{r-1}$ and $\tuple \mu \in [0,1]^r$
for $r \ge 3$, Mossel \cite{mossel} defines
\[
\Gamma_{\tuple \rho}(\mu_1, \dots, \mu_r) \coloneqq \Gamma_{\rho_1}\left(\mu_1, \Gamma_{\rho_2, \dots, \rho_{r-1}}(\mu_2, \dots, \mu_r)\right).
\]
Note that this multivariate $\Gamma_{\tuple \rho}$ is also increasing in every $\mu_i$.
\begin{observation}\label{obs:recurs-positive-gamma}
    \sloppy If $\mu_1, \dots, \mu_r > 0$ and $\rho_1, \dots, \rho_{r-1} < 1$,
    then $\Gamma_{\tuple \rho}(\mu_1, \dots, \mu_r) > 0$.
\end{observation}
\begin{proof}
    We prove it by induction. For the base case $r = 2$, recall from
    \Cref{obs:positive-gamma} that
    \[
    \Gamma_\rho(\mu_1, \mu_2) \ge \Gamma_\rho(\mu, \mu) > 0
    \]
    where $\mu \coloneqq \min\{\mu_1, \mu_2\} > 0$. For the inductive step, we
    have
    \[
    \Gamma_{\tuple \rho}(\mu_1, \dots, \mu_r) \coloneqq \Gamma_{\rho_1}\left(\mu_1, \Gamma_{\rho_2, \dots, \rho_{r-1}}(\mu_2, \dots, \mu_r)\right) > 0
    \]
    because $\Gamma_{\rho_2, \dots, \rho_{r-1}}(\mu_2, \dots, \mu_r) > 0$ by the inductive assumption.
\end{proof}

The rest of this section is devoted to proving~\cref{th:mossel}.
Its key ingredient are the powerful noise inequalities discovered by
Mossel~\cite{mossel}; we prove a slightly different variant of his result.
 
\begin{theorem}[Based on {\cite[Proposition 6.4]{mossel}}]\label{th:original-mossel}
    Let $\Omega$ be a distribution over $\prod_{i=1}^r A_i$ such that
    \[
    \rho \coloneqq \max_j \rho\left( A_j, \prod_{i \neq j} A_i; \Omega \right) < 1.
    \]
    For every $\varepsilon > 0$ there exists $k, \tau > 0$ such that\footnote{$k \to \infty$ and $\tau \to 0$ as $\rho \to 1$.} if $\left\{f_i : A_i^n \to [0,1]\right\}_{i \in [r]}$ are all $(k, \tau)$-regular w.r.t. $\Omega$, then
    \[
    \Gamma_{\tuple{\rho}} \left( \E_{\Omega_1^n}[f_1], \dots, \E_{\Omega_r^n}[f_r] \right) - \varepsilon \le \E_{\transpose{\bar x^{(1)}}{\bar x^{(r)}} \sim \Omega^n}\left[ \prod_{i=1}^r f_i(\bar x^{(i)}) \right].
    \]
\end{theorem}
We note that the only difference between this statement and \cite[Proposition 6.4]{mossel} is the definition of $\tuple\rho$.
Namely, Mossel defines it as
\[
\rho_i = \rho\left( A_1 \times \dots \times A_i, A_{i+1} \times \dots \times A_r \right)
\]
while we relax it to just $A_i$ on the left-hand side.
\begin{proof}
    We prove this using induction.
    The base case $r=2$ was proved by Mossel \cite[Proposition 6.3]{mossel}.
    Suppose $r > 2$.
    \cite[Lemma 6.2]{mossel} states that it suffices to prove the claim for the noised functions $\tilde f_i = T f_i$ as long as the noise parameter $\gamma$ is sufficiently small.
    By \cref{lem:low-degree-and-noise}, every $\tilde f_i$ is $(k, \tau)$-regular w.r.t. $\Omega$.
    Define $g : (A_2 \times \ldots \times A_r)^n \to [0,1]$ as
    \[
    (\bar x^{(2)}, \dots, \bar x^{(r)}) \mapsto \tilde f_2(\bar x^{(2)}) \cdot \ldots \cdot \tilde f_r(\bar x^{(r)}).
    \]
    By the inductive assumption, we have
    \[
    \Gamma_{\tuple{\rho}_{2:r}} \left( \E_{\Omega_2^n}[f_2], \dots, \E_{\Omega_r^n}[f_r] \right) - \varepsilon/2 \le \E_{\transpose{\bar x^{(1)}}{\bar x^{(r)}} \sim \Omega^n}\left[ \prod_{i=2}^r \tilde f_i(\bar x^{(i)}) \right] = \E[g].
    \]
    as long as $k$ is sufficiently large and $\tau$ is sufficiently small.

    Next, we are going to apply \cref{th:original-mossel} to $\tilde f_1$ and $g$.
    Let $k', \tau'$ be the values given by \cref{th:original-mossel} for $r=2$ and $\varepsilon/2$.
    We need to ensure that both $\tilde f_1$ and $g$ are $(k', \tau')$-regular.
    Indeed, by \cite[Lemma 6.5]{mossel}
    \[
    \Inf{g, j} \le r\sum_{i=2}^r \Inf{\tilde f_i, j} \le r \sum_{i=2}^r \left( \lowInf{f_i, j, k} + (1-\gamma)^{2k} \right) \le r^2 \left( \tau + (1-\gamma)^{2k} \right)
    \]
    where we use \cref{lem:low-degree-and-noise} and the regularity of $f_i$'s.
    Thus we pick $k \ge k'$ and $\tau \le \tau'$ such that the last expression is at most $\tau'$ (note that $r$ and $\gamma$ are fixed constants).
    Clearly $f_1$ is also $(k', \tau')$-regular, and, in turn, so is $\tilde f_1$ by \cref{lem:low-degree-and-noise}.

    Finally, \cref{th:original-mossel} applied to $\tilde f_1$ and $g$ gives
    \[
    \Gamma_{\rho_1}\left(\E_{\Omega_1^n}[f_1], \Gamma_{\tuple{\rho}_{2:r}} \left( \E_{\Omega_2^n}[f_2], \dots, \E_{\Omega_r^n}[f_r] \right) - \varepsilon/2 \right) - \varepsilon/2 \le \E_{\transpose{\bar x^{(1)}}{\bar x^{(r)}} \sim \Omega^n}\left[ \prod_{i=1}^r \tilde f_i(\bar x^{(i)}) \right].
    \]
    Since $\Gamma$ is 1-Lipschitz in each of its coordinates, the left-hand side is at most
    \[
    \Gamma_{\rho_1}\left(\E_{\Omega_1^n}[f_1], \Gamma_{\tuple{\rho}_{2:r}} \left( \E_{\Omega_2^n}[f_2], \dots, \E_{\Omega_r^n}[f_r] \right) \right) - \varepsilon = \Gamma_{\tuple{\rho}} \left( \E_{\Omega_1^n}[f_1], \dots, \E_{\Omega_r^n}[f_r] \right) - \varepsilon
    \]
    which concludes the proof.
\end{proof}

We are in a position to prove \cref{th:mossel}. We state it again, for the
reader's convenience.
\mossel*
\begin{proof}
    We are going to apply \cref{th:original-mossel} with $A_1 = \cdots = A_r = A$ and $f_1 =
    \cdots = f_r = g$. Fix $\delta > 0$.

    Since the support of $\Omega$ is BKLM-connected, we deduce from \cref{lem:connected-correlation} that the
    correlation
    \[
    \rho\left( A_j, \prod_{i \neq j} A_i; \Omega \right) = \rho < 1
    \]
    for any $j \in [r]$. Hence the assumptions of \cref{th:original-mossel} are
    satisfied.
    Furthermore, \cref{obs:subset-correlation} implies that $\tuple{\rho} \in [0, 1)^{r-1}$.
    Therefore $\Gamma_{\tuple{\rho}}(\delta, \dots,
    \delta) > 0$ by \cref{obs:recurs-positive-gamma}.

    Pick any $0 < \varepsilon' < \Gamma_{\tuple{\rho}}(\delta, \dots,
    \delta)$, and let $k, \tau > 0$ be the values given by \cref{th:original-mossel} for $\varepsilon'$; put $\varepsilon \coloneqq \Gamma_{\tuple{\rho}}(\delta, \dots, \delta) - \varepsilon' > 0$.
    For any $(k, \tau)$-regular $g \colon A^n \to \{0,1\}$, \cref{th:original-mossel} gives
    \[
    \varepsilon = \Gamma_{\tuple{\rho}} \left( \delta, \dots, \delta \right) - \varepsilon' \le \Gamma_{\tuple{\rho}} \left( \E_{\Omega_1}[g], \dots, \E_{\Omega_r}[g] \right) - \varepsilon' \le \E_{\transpose{\bar x^{(1)}}{\bar x^{(r)}} \sim \Omega^n}\left[ \prod_{i=1}^r g(\bar x^{(i)}) \right]
    \]
    by monotonicity of $\Gamma$, as required.
\end{proof}

\end{document}